\documentclass[letterpaper,leqno,11pt]{article}
\usepackage{amstext}
\usepackage{amsthm}
\usepackage{amsfonts}
\usepackage{amsmath}
\usepackage{latexsym}
\usepackage{amscd}
\usepackage{amssymb}
\usepackage{amsmath}

\numberwithin{equation}{section}

\newtheorem{theorem}{Theorem}[section]
\newtheorem{proposition}{Proposition}[section]
\newtheorem{lemma}{Lemma}[section]
\newtheorem{corollary}{Corollary}[section]

\theoremstyle{definition}
\newtheorem{definition}{Definition}[section]

\newtheorem{remark}{Remark}[section]

\newtheorem{conjecture}{Conjecture}[section]

\newcommand{\Q}{\mathbb{Q}}
\newcommand{\Z}{\mathbb{Z}}
\newcommand{\C}{\mathbb{C}}
\newcommand{\even}{{\bar{0}}}
\newcommand{\odd}{{\bar{1}}}
\newcommand{\Oth}{\mathbf{O}}
\newcommand{\g}{\mathfrak{g}}

\begin{document}

\title{Automorphisms and twisted forms of the $N=1,2,3$ Lie conformal superalgebras}

\author{\textsc{Zhihua Chang and Arturo Pianzola}}

\maketitle

\begin{abstract}
We classify the  $N=1,2,3$ superconformal Lie algebras of
Schwimmer and  Seiberg by means of differential non-abelian
cohomology, and describe the general philosophy behind this
new technique. The structure of the group (functor) of
automorphisms of the corresponding Lie conformal
superalgebra is a key ingredient of the proof.
\end{abstract}

\section{Introduction}

The purpose of this paper is two-fold.

(1) To bring to the attention of the physics community new
ideas, which for lack of a better word we will refer to as
differential non-abelian cohomology, that have recently
been used  to establish some deep results in infinite
dimensional Lie theory  (see, for example,
\cite{CGP,GP1,GP2,GP3,GP4,P1,P2}).  These powerful methods,
based on the seminal work of Demazure and Grothendieck on
reductive groups schemes, torsors and descent \cite{D,SGA1,SGA3}, can be adapted to the study of Lie conformal
superalgebras as explained in \cite{KLP}  by what amounts
to formally replacing the given base scheme,
$\rm{Spec}(\C[t^{\pm 1}])$ in our case, by  a differential
scheme.

(2) To complete the classification of the $N = 1,2,3,4$
superconformal Lie algebras by providing a {\it uniform}
proof of the cases $N = 1,2,3$, exploiting the fact  that
the relevant Lie conformal superalgebras used as base
objects for the twisted loop construction can be described
in terms of exterior algebras. We also provide a precise
argument that describes the passage from Lie conformal
superalgebras to their corresponding Lie superalgebras.
(The cases $N= 2,4$ were done  in \cite{KLP} at the Lie
conformal superalgebra level only, and by \textit{ad hoc}
methods). These Lie superalgebras are more relevant to
physics, where they are commonly referred to as
superconformal Lie algebras.

$N=1,2,3,4$ superconformal Lie algebras are a class of
infinite dimensional Lie superalgebras which plays
important roles in both mathematics and physics. They are
closed related to the twisted loop Lie conformal
superalgebra based on a complex Lie conformal superalgebra.
They were first introduced in \cite{K} as the affinization
of Lie conformal superalgebras over the complex numbers and
then realized as differential Lie conformal algebra in
\cite{KLP}. Based on the point of view taken in \cite{KLP},
a twisted loop Lie conformal superalgebra has both a
complex conformal superalgebra structure and an
$\mathcal{R}$-Lie conformal superalgebra structure where
$\mathcal{R}=(\C[t^{\pm1}],\frac{d}{dt}).$ While the
complex structure is of interest in physics, it is the
$\mathcal{R}$-structure that allows us to introduce
cohomological methods.

The twisted loop Lie conformal superalgebra construction
just mentioned is highly reminiscent of the way in which
the affine Kac--Moody Lie algebras, which a priori are
defined by generators and relations, are explicitly
realized in \cite{KK}.   This classification and
construction has recently been established by means of
non-abelian \'etale cohomology \cite{P1}. Because of the
very special nature of the algebraic fundamental group of
the ring $\C[t^{\pm 1}]$, the loop algebras based on a
finite-dimensional simple Lie algebra $\g$ are
parameterized by the conjugacy classes of the finite group
of symmetries of the corresponding Coxeter--Dynkin diagram.
In particular, if the automorphism group of $\g$ is {\it
connected}, then {\it all loop algebras based on $\g$ are
trivial}, i.e., isomorphic to $\g \otimes_{\C} \C[t^{\pm
1}].$ As the work of Schwimmer and Seiberg had shown,
herein lies a story.

A loop inspired procedure  as a mean of realizing  the
$N$-superconformal Lie algebras is given in the
striking\footnote{Certainly to mathematicians.} \cite{SS}.
In the $N = 2$ case, the authors show that all the objects
in what was until then thought to be an ``infinite''  family
of superconformal algebras were in fact isomorphic, and
that only two non-isomorphic classes of superconformal Lie
algebras existed.\footnote{Strictly speaking their work
shows that {\it at most} two classes exist. We entertain no
doubts that it was obvious to the authors that the two
classes were different.} This agrees precisely with the
what cohomological point of view for algebras would predict
since in the $N=2$ case the automorphism group $\Oth_2$ of
the corresponding Lie conformal superalgebra has two
connected components. By contrast, Schwimmer and Seiberg
put forward an infinite family of non-isomorphic $N=4$
superconformal Lie algebras\,$\ldots$\,even though in this case the
automorphism group is connected! How a connected group of
automorphism could lead to an infinite family of loop
objects, and which cohomology can be used to determine
them, is explained in \cite{KLP}. The crucial idea is that,
unlike the case of algebras, the base ring $\C[t^{\pm 1}]$
does not contain enough information to geometrically
measure superconformal Lie algebras. It must be replaced by
the complex {\it differential ring} $(\C[t^{\pm 1}],
\delta)$ where $\delta = \frac{d}{dt}.$ One could in fact
take an arbitrary derivation. For example, the case of
$\delta = 0$ leads to the classification of ``current
algebras''.

Forms of a given $\mathcal{R}$-Lie conformal superalgebra
$(\mathcal{A},\partial)$ are classified in terms of the
non-abelian cohomology pointed set
$H^1\big(\mathcal{R},\mathbf{Aut}(\mathcal{A})\big)$ as
explained in \cite{KLP}. In this paper, we focus on the
$N=1,2,3$ Lie conformal superalgebras. The classification
of their twisted loop Lie conformal superalgebras will be
carried by explicitly computing their automorphism group
functors and the corresponding non-abelian cohomology sets.

Concretely, we first review the general theory of
differential Lie conformal algebras developed in \cite{KLP}
in Section  2. Then classification will be accomplished
along the following lines. In Section  3, we compute the
automorphism group functor of the $N=1,2,3$ Lie conformal
superalgebras $\mathcal{K}_N$. The construction is quite
explicit and we believe of interest to physicists. In
Section 4, we complete the classification of forms of
$\mathcal{K}_N\otimes_\C\mathcal{R}$ up to isomorphism of
$\mathcal{R}$-Lie conformal algebras by compute the
corresponding nonabelian cohomology set. Centroid
considerations are then used to show that no information is
lost in the passage from $\mathcal{R}$ to $\C.$ Finally in
Section  5, we pass from twisted loop Lie conformal
superalgebras to their corresponding superconformal Lie
algebras, and again show that no collapse occurs to the
isomorphism classes. This completes the
classification.\bigskip

\noindent\textbf{Notation:} $\Z,\Bbb N,\Q~\mathrm{and}~\C $
denote the integers, non-negative integers, rational
numbers and complex numbers respectively. $\mathbf{i}$ will
denote $\sqrt{-1}$. For a cycle $(ijl)$, $\epsilon_{ijl}$
denotes its sign.

$\mathcal{R},\mathcal{S}_m,\widehat{\mathcal{S}}$ always
denote the complex differential algebras
$(\C[t^{\pm1}],\delta_t)$, $(\C[t^{\pm1}],\delta_t),
(\C[t^q,q\in\Q], \delta_t)$, where $\delta_t=\frac{d}{dt}$.

\section{Preliminaries}\label{secP}

In this section, we will review the general theory of
differential Lie conformal superalgebras developed in
\cite{KLP}.\footnote{This concept is not to be confused
with that of superconformal Lie algebra as it appears in
the physics literature. The two concepts are related, as we
will explain in Remark \ref{terminology}.}

Motivated by the affinization of complex Lie conformal
superalgebras defined in \cite{K}, Lie conformal
superalgebras over an arbitrary  differential rings were
introduced in \cite{KLP}. In this paper, we only consider
the Lie conformal superalgebras over a complex differential
ring $\mathcal{D}.$ Thus $\mathcal{D}=(D,\delta)$ is a pair
consisting of a commutative associative algebra $D$ over
$\C$ and a $\C$-linear derivation $\delta:D\rightarrow D$.
Complex differential rings form a category, where a morphism
$f:(D,\delta)\rightarrow(D',\delta')$ is a $\C$-linear ring
homomorphism $f:D\rightarrow D$ such that
$f\circ\delta=\delta'\circ f$.

\begin{definition}[{(\cite{KLP}, Definition 1.3)}]
Let $\mathcal{D}=(D,\delta)$ be a complex differential
ring, a $\mathcal{D}$-Lie conformal superalgebra is a
triple
$(\mathcal{A},\partial_{\mathcal{A}},(-_{(n)}-)_{n\in\Bbb
N})$ consisting of
\begin{quote}
\begin{enumerate}
\item[$\mathrm{(i)}$] a $\Z/2\Z$-graded $D$-module $\mathcal{A}=\mathcal{A}_\even\oplus\mathcal{A}_{\odd}$, \item[$\mathrm{(ii)}$] $\partial_{\mathcal{A}}\in\mathrm{End}_{\C}(\mathcal{A})$ stabilizing the even and odd parts of $\mathcal{A}$, and
\item[$\mathrm{(iii)}$] a $\C$-bilinear product $(a,b)\mapsto a_{(n)}b, a,b\in\mathcal{A}$ for each $n \in \Bbb N$,
\end{enumerate}
\end{quote}
satisfying the following axioms for $r\in D, a,b,c\in\mathcal{A}, m,n\in\Bbb N$:
\begin{quote}
\begin{enumerate}
\item[$\mathrm{(CS0)}$] $a_{(n)}b=0$ for $n\gg0$,
\item[$\mathrm{(CS1)}$] $\partial_{\mathcal{A}}(a)_{(n)}b=-na_{(n-1)}b$ and $a_{(n)}\partial_{\mathcal{A}}(b)=\partial_{\mathcal{A}}(a_{(n)}b)+na_{(n-1)}b$,
\item[$\mathrm{(CS2)}$] $\partial_{\mathcal{A}}(ra)=r\partial_{\mathcal{A}}(a)+\delta(r)a$,
\item[$\mathrm{(CS3)}$] $a_{(n)}(rb)=r(a_{(n)}b)$ and $(ra)_{(n)}b=\sum_{j\in\Bbb N}\delta^{(j)}(r)(a_{(n+j)}b)$,
\item[$\mathrm{(CS4)}$] $a_{(n)}b=-p(a,b)\sum_{j\in\Bbb N}(-1)^{j+n}\partial_{\mathcal{A}}^{(j)}(b_{(n+j)}a)$, and
\item[$\mathrm{(CS5)}$] $a_{(m)}(b_{(n)}c)=\sum_{j=0}^m\binom{m}{j}(a_{(j)}b)_{(m+n-j)}c+p(a,b)b_{(n)}(a_{(m)}c)$,
\end{enumerate}
\end{quote}
where $\delta^{(j)}=\frac{1}{j!}\delta^j,
\partial_{\mathcal{A}}^{(j)}=\frac{1}{j!}\partial_{\mathcal{A}}^{j},
j\in\Bbb N$ and $p(a,b)=(-1)^{p(a)p(b)}$, $p(a)$ (resp.
$p(b)$) is the parity of $a$ (resp. $b$).
\end{definition}

We also use the $\lambda$-bracket convention  to simplify notation. Recall then that
\[
[a_\lambda b]=\sum\limits_{n\in\Bbb N}\lambda^{(n)}a_{(n)}b,
\]
where $\lambda$ is a variable, $\lambda^{(n)}=\frac{1}{n!}\lambda^n$ and $a,b \in \mathcal{A}$.

If we consider  $\C$  as a differential ring with
$\delta=0$, the $\C$-Lie conformal superalgebras defined
above coincide with the usual definition found in \cite{K}.

Let $\mathcal{D}=(D,\delta)\rightarrow\mathcal{D}'=(D',\delta')$
be an extension of  differential rings. A
$\mathcal{D}'$-Lie conformal superalgebra can be viewed as
a $\mathcal{D}$-Lie conformal superalgebra in the natural
way. Conversely, for a $\mathcal{D}$-Lie conformal
superalgebra $(\mathcal{A},\partial_{\mathcal{A}})$, there
is a $\mathcal{D'}$-Lie conformal superalgebras structure
on $\mathcal{A}\otimes_{\mathcal{D}}\mathcal{D}'$ given by
\[
\partial_{\mathcal{A}\otimes_{\mathcal{D}}\mathcal{D}'}(a\otimes r)=\partial_{\mathcal{A}}(a)\otimes r+a\otimes\delta(r),
\]
for $a\in\mathcal{A}, r\in D'$, and
\[
(a\otimes f)_{(n)}(b\otimes g)=\sum\limits_{j\in\Bbb N}(a_{(n+j)}b)\otimes \delta^{(j)}(f)g,
\]
for $a,b\in\mathcal{A}, f,g\in D, n\in\Bbb N.$ With base change defined we make the following definition.

\begin{definition}
Let $\mathcal{D}\rightarrow\mathcal{D'}$ be an extension of
differential rings and $(\mathcal{A},\partial)$ a
$\mathcal{D}$-Lie conformal superalgebra, a
$\mathcal{D}'/\mathcal{D}$-form of $\mathcal{A}$ is a
$\mathcal{D}$-Lie conformal superalgebra $\mathcal{L}$ such
that
\[
\mathcal{L}\otimes_{\mathcal{D}}\mathcal{D}'\cong\mathcal{A}\otimes_{\mathcal{D}}\mathcal{D}'
\]
as $\mathcal{D}'$-conformal superalgebras.
\end{definition}

If the extension $\mathcal{D}'/\mathcal{D}$ is a faithfully
flat extension, i.e., if the ring extension $D\rightarrow
D'$ is faithfully flat, the set of isomorphism classes of
$\mathcal{D}'/\mathcal{D}$-forms of a $\mathcal{D}$-Lie
conformal superalgebra $\mathcal{A}$ is identified with the
non-abelian \v{C}ech cohomology point set
$H^1\big(\mathcal{D}'/\mathcal{D},\mathbf{Aut}(\mathcal{A})\big)$,
where $\mathbf{Aut}(\mathcal{A})$ is the group functor from
the category of differential extensions of $\mathcal{D}$ to
the category of groups which assigns to an extension
$\mathcal{D}'$ of $\mathcal{D}$ the group
$\mathbf{Aut}(\mathcal{A})(\mathcal{D}')$ of automorphisms
of the $\mathcal{D}'$-Lie conformal superalgebra
$\mathcal{A}\otimes_{\mathcal{D}} \mathcal{D}'$ (see
Theorem~2.16 of~\cite{KLP}).

\begin{remark}\label{grpfunrmk}
The $\mathcal{D}$-group functor $\mathbf{Aut}(\mathcal{A})$
plays a key role in the classification of
$\mathcal{D}'/\mathcal{D}$-forms of a given
$\mathcal{D}$-Lie conformal superalgebra $\mathcal{A}$. Let
$\mathcal{L}$ be such a form. We make the following
definition for future use. We say that a subgroup functor
$\mathcal{F}$  of $\mathbf{Aut}(\mathcal{L})$ is {\it
representable} if there exists a scheme $\mathfrak{X}$ over
${\rm Spec}(D)$ such that $\mathcal{F}(\mathcal{E}) \simeq
{\rm Hom}_{D\text{-sch}}({\rm Spec}(E), \mathfrak{X})$ for
every differential extension $\mathcal{E} = (E ,
\delta_{\mathcal{E}})$ of $\mathcal{D}$, where the
identifications  are ``functorial on $\mathcal{E}$''.
Recall that by Yoneda's correspondence if $\mathfrak{X} =
{\rm Spec}(A)$ is affine, then
\[
{\rm Hom}_{D\text{-sch}}({\rm Spec}(E), \mathfrak{X}) \simeq {\rm Hom}_{D\text{-alg}}(A, E).
\]
\end{remark}

In this paper, we are mainly interested in the twisted loop
Lie conformal superalgebra
$\mathcal{L}(\mathcal{A},\sigma)$ based on a complex Lie
conformal superalgebra $(\mathcal{A},\partial)$ with
respect to an automorphism
$\sigma:\mathcal{A}\rightarrow\mathcal{A}$ of finite order
$m$. Recall that this is defined by
\[
\mathcal{L}(\mathcal{A},\sigma)=\bigoplus\limits_{n\in\Z}\mathcal{A}_{n}\otimes t^{n/m}\subseteq\mathcal{A}\otimes_{\C}\mathcal{S}_m,
\]
where $\mathcal{A}_n=\{a\in\mathcal{A}|\sigma(a)=\zeta_m^na\}, n\in\Z$, $\zeta_m = e^{\frac{2\pi\bf{i}}{m}}$  and $\mathcal{S}_m=(\C[t^{\pm\frac{1}{m}}],\delta_t)$.

The twisted loop algebra $\mathcal{L}(\mathcal{A},\sigma)$
is not only a complex Lie conformal superalgebra but also
an $\mathcal{R}$-Lie conformal superalgebra for
$\mathcal{R}=(\C[t^{\pm1}],\delta_t)$. It is {\it
trivialized} by the extension $\mathcal{S}_m/\mathcal{R}$,
hence also by $\widehat{\mathcal{S}}=\lim\limits_{\longrightarrow}(\mathcal{S}_m,\delta_t).$\footnote{The
introduction of $\widehat{\mathcal{S}}$ is a useful
artifice that allows us to compare loop algebras of
automorphisms of arbitrary order all at once.} This means
that after applying the base change $\mathcal{R} \to
\widehat{\mathcal{S}}$ our object ``splits''. To be
precise,
\[
\mathcal{L}(\mathcal{A},\sigma) \otimes_{\mathcal{R}} \widehat{\mathcal{S}} \simeq \mathcal {A} \otimes_\C \widehat{\mathcal{S}} \simeq (\mathcal {A} \otimes_\C \mathcal{R})\otimes_{\mathcal {R}} \widehat{\mathcal{S}},
\]
where all of the above are isomorphisms of
$\widehat{\mathcal{S}}$-Lie conformal superalgebras. In
other words, $\mathcal{L}(\mathcal{A},\sigma)$ is {\it an}
$\widehat{\mathcal{S}}/\mathcal{R}$-{\it form of}
$\mathcal{A}\otimes_\C\mathcal{R}$.\footnote{The spirit of
this construction is not unlike that of vector bundles,
which locally look like a well-understood object, namely
the trivial bundle. The difference is in the meaning of
``locally'', which is to be understood not for the usual
topology of $\C^\times$, but in a differential version of
Grothendieck's {\it fppf}-topology on ${\rm Spec}(R).$ The
trivialization takes place by using ``only one open set'',
namely ${\rm Spec}(S).$} We will denote in what follows the
$\mathcal{R}$-conformal superalgebra
$\mathcal{A}\otimes_\C\mathcal{R}$ by
$\mathcal{A}_{\mathcal{R}}.$

Since $\widehat{\mathcal{S}}/\mathcal{R}$ is faithfully
flat, the set of $\mathcal{R}$-conformal isomorphism
classes of $\widehat{\mathcal{S}}/\mathcal{R}$-forms of
$\mathcal{A}\otimes_\C\mathcal{R}$ is identified with the
non-abelian \v{C}ech cohomology pointed set
$H^1\big(\widehat{\mathcal{S}}/\mathcal{R},\mathbf{Aut}(\mathcal{A}_{\mathcal{R}})\big)$.

Furthermore, if $\mathcal{A}_{\mathcal{R}}$ satisfies the following finiteness condition:
\begin{enumerate}
\leftskip5pt
\item[\textbf{(Fin)}] There exist finitely many elements $a_1,\ldots,a_n\in\mathcal{A}_{\mathcal{R}}$ such that the set $\{\partial_{\mathcal{A}}^\ell(ra_i)|r\in R,\ell\geq0\}$ spans $\mathcal{A}_{\mathcal{R}}$.
\end{enumerate}
(which holds in all of the cases we are interested in), the
$H^1\big(\widehat{\mathcal{S}}/\mathcal{R},\mathbf{Aut}(\mathcal{A}_{\mathcal{R}})\big)$
we are after can be identified with the non-abelian
continuous cohomology
$H^1\big(\pi_1(R),\mathbf{Aut}(\mathcal{A}_{\mathcal{R}})(\widehat{\mathcal{S}})\big)$,
where $\pi_1(R)$ is the algebraic fundamental group of
$\mathrm{Spec}(R)$ at the geometric point
$\mathrm{Spec}\big(\overline{\C(t)}\big)$ (See \cite{KLP},
Proposition~2.29).

In particular for a complex Lie conformal superalgebra
$\mathcal{A}$ we can consider $\mathbf{Aut}(\mathcal{A})$,
a functor that can be evaluated at any differential ring to
yield, in doing so,  a group. The group functor
$\mathbf{Aut}(\mathcal{A}_{\mathcal{R}})$ considered above
is nothing but the restriction of
$\mathbf{Aut}(\mathcal{A})$ to the category of differential
extensions of $\mathcal{R}$ (such extensions are by
definition naturally endowed with a  differential ring
structure). In particular $\mathbf{Aut}(\mathcal{A})(\mathcal{\widehat{S}}) =
\mathbf{Aut}(\mathcal{A}_{\mathcal{R}})(\mathcal{\widehat{S}})$.
This makes the \hbox{cohomology} set
$H^1\big(\widehat{\mathcal{S}}/\mathcal{R},\mathbf{Aut}(\mathcal{A}_{\mathcal{R}})\big)$
computable, which yields the \hbox{classification} of
$\widehat{\mathcal{S}}/\mathcal{R}$-forms of
$\mathcal{A}\otimes_\C\mathcal{R}$ {\it up to isomorphisms
of $\mathcal{R}$-Lie conformal superalgebras.} Then the
centroid trick can be brought in to obtain the passage from
the classification of $\mathcal{R}$-Lie conformal
superalgebras to  that of complex Lie conformal
superalgebras.

In this paper, we consider twisted loop conformal
superalgebras base on complex Lie conformal superalgebras
$\mathcal{K}_N, N=1,2,3$ which are described in \cite{K} as
follows:

Let $\Lambda(N)$ be the Grassmann algebra over $\C$ in $N$
variables $\xi_1,\ldots,\xi_N$. $\Lambda(N)$ has a
$\Z/2\Z$-grading given by setting each $\xi_i,
i=1,\ldots,N$ to be odd. This gives $\Lambda(N)$ the
structure of a complex superalgebra. $\Lambda(N)$ also has
a $\Z$-grading in which each $\xi_i, i=1,\ldots,N$ has
degree $1$. If $f$ is homogeneous with respect to the
$\Z$-grading, we denote by $|f|$ the degree of $f$. The
complex Lie conformal superalgebra $\mathcal{K}_N$ is the
naturally  $\Z/2\Z$-graded complex vector space
\[
\mathcal{K}_N=\C[\partial]\otimes_\C\Lambda(N)
\]
equipped with the $n$th product for each $n\in\Bbb N$ given by
\begin{align*}
f_{(0)}g&=\left(\frac{1}{2}|f|-1\right)\partial\otimes fg+\frac{1}{2}(-1)^{|f|}\sum_{i=1}^N(\partial_i f)(\partial_i g),\\
f_{(1)}g&=\left(\frac{1}{2}(|f|+|g|)-2\right)fg,\\
f_{(n)}g&=0, \quad n\geqslant2,
\end{align*}
where $f,g\in\Lambda(N)$ are homogenous elements with
respect to the $\Z$-grading, and $\partial_i$ is the
derivative with respect to $\xi_i$, $i=1,\ldots,N$.

\begin{remark}\label{terminology}
Every twisted loop Lie conformal superalgebra
$\mathcal{L}(\mathcal{A},\sigma)$ based on a complex Lie
conformal superalgebra $(\mathcal{A},\sigma)$ naturally
corresponds to a complex Lie superalgebra:
\[
\mathrm{Alg}(\mathcal{A},\sigma)=\mathcal{L}(\mathcal{A},\sigma)/\left(\partial+\frac{d}{dt}\right)\mathcal{L}(\mathcal{A},\sigma),
\]
where  the Lie bracket comes from the zeroth product of
$\mathcal{L}(\mathcal{A},\sigma).$ The classification of
loop algebras based on a complex Lie conformal superalgebra
$\mathcal{A}$ is thus related, in an essential and
meaningful way, to that of their corresponding Lie
superalgebras. The Lie superalgebras corresponding to the
twisted loop superalgebras base on complex conformal
superalgebras $\mathcal{K}_N, N=1,2,3$ are exactly the
$N=1,2,3$ {\it superconformal Lie algebras} described in
\cite{SS}. We shall study these algebras in Section 5.
\end{remark}

\section{Automorphism groups of the $N=1,2,3$ Lie conformal superalgebras}\label{sec2}

The underlying vector space  of the complex Lie conformal
superalgebras $\mathcal{A}$ that we are interested are all
of the form
\[
\mathcal{A}=\C[\partial]\otimes_\C V
\]
for some complex vector superspace $V.$  In particular,
$\mathcal{A}$  is a free $\C[\partial]$-supermodule, and
has a natural $\Z$-grading (as a superspace) with
$\mathcal{A}_n = \C\partial^n \otimes_\C V.$ Let $\phi \in
{\bf Aut}(\mathcal{A}).$ Since $\phi$ commutes with the
action of $\partial$ to ask that $\phi$  preserve the
$\Z$-grading is the same than to ask that $\phi(1 \otimes
V) \subseteq 1 \otimes V.$\footnote{The inclusion is
necessarily an equality since $\phi$ is surjective.} This
leads us to consider a natural subgroup functor
$\mathbf{GrAut}(\mathcal{A})$ of
$\mathbf{Aut}(\mathcal{A})$ defined by
\[
\mathbf{GrAut}(\mathcal{A})(\mathcal{D})=\left\{\phi\in\mathbf{Aut}(\mathcal{A})(\mathcal{D})\middle|
\phi\left(\C\otimes_\C V\otimes_\C D\right)\subseteq\C\otimes_\C V\otimes_\C D\right\},
\]
for all complex differential ring $\mathcal{D}=(D,\delta).$
The preservation of degrees is a natural condition to
impose, and it carries also a physical meaning.

For $\mathcal{K}_N, N=1,2,3$, the subgroup functor is
defined by
\begin{align*}
&\mathbf{GrAut}(\mathcal{K}_N)(\mathcal{D})\\&\qquad{}=\left\{\phi\in\mathbf{Aut}(\mathcal{K}_N)(\mathcal{D})\middle|
\phi\left(\C\otimes_\C \Lambda(N)\otimes_\C D\right)\subseteq\C\otimes_\C\Lambda(N)\otimes_\C D\right\},
\end{align*}
where $\mathcal{D}=(D,\delta)$ is an arbitrary complex
differential ring.

The next theorem gives a  precise characterization of the
group functor $\mathbf{GrAut}(\mathcal{K}_N)$, $N=1,2,3$.
We will also show that the functors
$\mathbf{GrAut}(\mathcal{K}_N)$ and
$\mathbf{Aut}(\mathcal{K}_N)$ are equal when evaluated at
complex differential rings that are integral domains, in
particular at $\widehat{\mathcal{S}}=(\C[t^q,q\in\Q],\frac{d}{dt}).$

To simplify our notation, we set
$\widehat{\partial}=\partial\otimes\mathrm{id}+\mathrm{id}\otimes\delta$,
which can be viewed as an operator on
$\mathcal{K}_N\otimes_{\C}\mathcal{D}$. And for convenience
we identify $f\in\Lambda(N)$ with its image $(1\otimes
f)\otimes1$ in $\mathcal{K}_N\otimes_{\C}\mathcal{D}$.

\begin{theorem}\label{auto-thm}
Let $\mathcal{D}=(D,\delta)$ be an arbitrary complex differential ring. For $N=1,2,3$, there is an isomorphism  of groups
\[
\iota_{\mathcal{D}}:\Oth_N(D)\overset{\sim}{\rightarrow}\mathbf{GrAut}(\mathcal{K}_N)(\mathcal{D}), \quad A=(a_{ij})_{N\times N}\mapsto \phi_A,
\]

\removelastskip\pagebreak

\noindent where
$\phi_A\in\mathbf{GrAut}(\mathcal{K}_N)(\mathcal{D})$ is
given by the following data:
\begin{itemize}
\item For $N=1$, $\phi_A(1)=1\quad\text{ and }\quad\phi_A(\xi_1)=\xi_1\otimes a_{11}.$
\item For $N=2$,
\begin{align*}
&\phi_A(1)=1+\xi_1\xi_2\otimes r,&&\phi_A(\xi_1)=\xi_1\otimes a_{11}+\xi_2\otimes a_{21},\\
&\phi_A(\xi_1\xi_2)=\xi_1\xi_2\otimes\det(A),&&\phi_A(\xi_2)=\xi_1\otimes a_{12}+\xi_2\otimes a_{22},
\end{align*}
where $\begin{pmatrix}0&r\\
-r&0\end{pmatrix}=2\delta(A)A^T$.
\item For $N=3$,
\begin{align*}
&\phi_A(1)=1+\sum\limits_{l=1}^3\epsilon_{mnl}\xi_m\xi_n\otimes r_l, &&\phi_A(\xi_j)=\sum\limits_{l=1}^3\xi_l\otimes a_{lj}+\xi_1\xi_2\xi_3\otimes s_j, \\
&\phi_A(\xi_1\xi_2\xi_3)=\xi_1\xi_2\xi_3\otimes\det(A),&&\phi_A(\xi_i\xi_j)=\epsilon_{ijl}\sum\limits_{l'=1}^3\epsilon_{mnl'}\xi_{m}\xi_{n}\otimes A_{l'l},
\end{align*}
$i,j=1,2,3, i\neq j$, where $A_{l'l}$ is the cofactor of $a_{l'l}$ in $A$ and
\begin{align*}
\begin{pmatrix}0&r_3&-r_2\\
-r_3&0&r_1\\
r_2&-r_1&0\end{pmatrix}&=2\delta(A)A^\mathrm{T},\\
\begin{pmatrix}0&s_3&-s_2\\
-s_3&0&s_1\\
s_2&-s_1&0\end{pmatrix}&=2(\det A)A^\mathrm{T}\delta(A).
\end{align*}
\end{itemize}
This construction is functorial on $\mathcal{D}.$ In particular $\Oth_N \simeq \mathbf{GrAut}(\mathcal{K}_N).$
\end{theorem}
\begin{proof}
For $A\in\Oth_N(D)$, the data given in the theorem defines
an $D$-module homomorphism
$\phi_A:\Lambda(N)\otimes_{\C}D\rightarrow\Lambda(N)\otimes_{\C}D$.
One establishes that $\phi_A$ is a module isomorphism by
explicitly checking that $\phi_{A^{-1}}$ is its inverse.
Then $\phi_A$ can be extended to a unique $D$-module
isomorphism
$\mathcal{K}_N\otimes_\C\mathcal{D}\rightarrow\mathcal{K}_N\otimes_\C
\mathcal{D}$, which is also denoted by $\phi_A$, such that
$\widehat{\partial}\circ\phi_A=\phi_A\circ\widehat{\partial}$.
A direct computation shows that
\[
\phi_A([f\otimes1_\lambda g\otimes1])=[\phi_A(f\otimes1)_\lambda\phi_A(g\otimes1)]
\]
for all $f,g\in\Lambda(N).$ Since
$\mathcal{K}_N=\C[\partial]\otimes_\C \Lambda(N)$ is a free
$\C[\partial]$-module, we deduce that
$\phi_A\in\mathbf{Aut}(\mathcal{K}_N)(\mathcal{D})$ by
\cite{KLP}, Lemma~3.1(ii). Moreover, we observe that
$\phi_A\in\mathbf{GrAut}(\mathcal{K}_N)(\mathcal{D})$ by
the definition of $\phi_A$.

The above yields a function
$\iota_{\mathcal{D}}:\Oth_N(D)\rightarrow\mathbf{GrAut}(\mathcal{K}_N)(\mathcal{D})
$ such that $A\mapsto\phi_A.$ For
$\phi_1,\phi_2\in\mathbf{Aut}(\mathcal{K}_N)(\mathcal{D})$,
$\phi_1=\phi_2$ if and only if $\phi_1(f)=\phi_2(f)$ for
all $f\in\Lambda(N)$ by \cite{KLP},
Lemma~3.1(i).\footnote{This useful result will be used
repeatedly in what follows without further reference.} In
particular, for $A,B\in\Oth_N(D)$, $\phi_{AB}(f)=
(\phi_A\circ\phi_B)(f)$ for all $f\in\Lambda(N)$, thus
$\phi_{AB}=\phi_A\circ\phi_B$, i.e., $\iota_{\mathcal{D}}$
is a group homomorphism. By the same reasoning $A=B$ if
$\phi_A(f)=\phi_B(f)$ for all $f\in\Lambda(N)$, so that
$\iota_{\mathcal{D}}$ is injective.

It remains to show that $\iota_{\mathcal{D}}$ is
surjective, namely that given an automorphism $\phi \in
\mathbf{GrAut}(\mathcal{K}_N)(\mathcal{D})$ there exists
$A\in\Oth_N(D)$ such that $\phi=\phi_A$.\vskip12pt

\textbf{Case $N=1$:} Since $\phi(\C\otimes\Lambda(1)\otimes D)\subseteq\C\otimes\Lambda(1)\otimes D$ and $\phi$ preservers the $\Z/2\Z$-grading of $\mathcal{K}_1\otimes_\C \mathcal{D}$, we may assume that
$\phi(\xi_1)=\xi_1\otimes r$ for some $r\in D$.

Since $\phi^{-1}\in\mathbf{GrAut}(\mathcal{K}_N)(\mathcal{D})$,
we may also assume that $\phi^{-1}(\xi_1)=\xi_1\otimes s$
for some  $s\in D$. Then
$\phi\circ\phi^{-1}(\xi_1)=\xi_1\otimes rs=\xi_1\otimes1$.
Thus $rs=1$. i.e., $r$ is an unit in $D$.

We deduce from $\phi(\xi_1)=\xi_1\otimes r$ that
$\phi(1)=-2\phi(\xi_1)_{(0)}\phi(\xi_1)=1\otimes r^2$.
While, $\phi(1)_{(1)}\phi(1)=-2\phi(1)$ implies that
$r^4=r^2$. Since  $r$ is an unit in $D$ we obtain $r^2=1.$
In particular, $r\in\Oth_1(D)$ and
$\phi=\phi_{r}.$\vskip12pt

\textbf{Case $N=2$:} $\phi(\C\otimes\Lambda(2)\otimes D)\subseteq\C\otimes\Lambda(2)\otimes D$ and $\phi$ preservers the $\Z/2\Z$-grading of $\mathcal{K}_2\otimes_\C \mathcal{D}$ yield
\[
\phi(\xi_1)=\xi_1\otimes a_{11}+\xi_2\otimes a_{21},
\text{ and }
\phi(\xi_2)=\xi_1\otimes a_{12}+\xi_2\otimes a_{22},
\]
where $a_{ij}\in D,i,j=1,2.$

Let $A=(a_{ij})_{2\times 2}$. Since $\phi$ has an inverse
in $\mathbf{GrAut}(\mathcal{K}_N)(\mathcal{D})$, the matrix
$A$ is necessarily invertible. Now
\[
\phi(\xi_1\xi_2)=-\phi(\xi_1)_{(1)}\phi(\xi_2)=\xi_1\xi_2\otimes(a_{11}a_{22}-a_{21}a_{12})=\xi_1\xi_2\otimes\det(A).
\]
Choose $c,r\in D$ such that $\phi(1)=1\otimes c+\xi_1\xi_2\otimes r$. From $\phi(1)_{(1)}\phi(\xi_1\xi_2)=-\phi(\xi_1\xi_2)$ we deduce that $c\cdot\det(A)=\det(A)$. Since $A$ is invertible, $\det(A)$ is an unit in $D$ and therefore $c=1$.

Since $\phi(\xi_j)_{(0)}\phi(\xi_j)=-\frac{1}{2}\phi(1)$, we have
\[
a_{1j}^2+a_{2j}^2=1, \quad\text{ and }\quad r=2\big(\delta(a_{1j})a_{2j}-a_{1j}\delta(a_{2j})\big),\quad j=1,2,
\]
while $\phi(\xi_1)_{(0)}\phi(\xi_2)=-\frac{1}{2}\widehat{\partial}\phi(\xi_1\xi_2)$
implies that $a_{11}a_{12}+a_{21}a_{22}=0$. Thus\break
$A=(a_{ij})\in\Oth_2(D)$ and
\begin{align*}
r&=\big(\delta(a_{11})a_{21}-a_{11}\delta(a_{21})\big)+\big(\delta(a_{12})a_{22}-a_{12}\delta(a_{22})\big)\\
&=\big(\delta(a_{11})a_{21}-a_{11}\delta(a_{21})\big)+\big(\delta(a_{12})a_{22}-a_{12}\delta(a_{22})\big)+\delta(a_{11}a_{21}+a_{12}a_{22})\\
&=2\big(\delta(a_{11})a_{21}+\delta(a_{12})a_{22}\big),
\end{align*}
i.e., $\begin{pmatrix}0&r\\-r&0\end{pmatrix}=2\delta(A)A$. It follows $\phi(f)=\phi_A(f)$ for all $f\in\Lambda(2)$. Hence,
$\phi=\phi_A.$\vskip12pt

\textbf{Case $N=3$:} By reasoning as above we may assume that
\begin{equation}
\phi(\xi_j)=\xi_1\otimes a_{1j}+\xi_2\otimes a_{2j}+\xi_3\otimes a_{3j}+\xi_1\xi_2\xi_3\otimes s_j,\label{eqN31}
\end{equation}
where $a_{ij}, s_j\in D, i,j=1,2,3$, and that the matrix  $A=(a_{ij})$ is invertible.
For $i\neq j$, we have
\begin{align}
\phi(\xi_i\xi_j)&=-\phi(\xi_i)_{(1)}\phi(\xi_j)\label{eqN32}\\
&=\xi_1\xi_2\otimes(a_{1i}a_{2j}-a_{2i}a_{1j})+\xi_2\xi_3\otimes(a_{2i}a_{3j}-a_{3i}a_{2j})\nonumber\\
&\quad +\xi_3\xi_2\otimes(a_{3i}a_{2j}-a_{2i}a_{3j})\nonumber\\
&=\epsilon_{ijk}(\xi_1\xi_2\otimes A_{3k}+\xi_2\xi_3\otimes A_{1k}+\xi_3\xi_1\otimes A_{2k}),\nonumber
\end{align}
where $A_{ij}$ is the cofactor of $a_{ij}$ in $A$. While,
\begin{align}
\phi(\xi_1\xi_2\xi_3)&=-2\phi(\xi_1\xi_2)_{(1)}\phi(\xi_3)\label{eqN33}\\
&=\xi_1\xi_2\xi_3\otimes(A_{13}a_{13}+A_{23}a_{23}+A_{33}a_{33})\nonumber\\
&=\xi_1\xi_2\xi_3\otimes\det(A).\nonumber
\end{align}

Write $\phi(1)=1\otimes c+\xi_1\xi_2\otimes
r_3+\xi_2\xi_3\otimes r_2+\xi_3\xi_1\otimes r_1$, $c,r_j\in
D,j=\break 1,2,3$. Then
$\phi(1)_{(1)}\phi(\xi_1\xi_2\xi_3)=-\frac{1}{2}\phi(\xi_1\xi_2\xi_3)$
yields $c\cdot\det(A)=\det(A)$. Thus $c=1.$ We will show
that $A\in\Oth_3(D)$. First,
$\phi(\xi_j)_{(0)}\phi(\xi_j)=-\frac{1}{2}\phi(1), j=\break 1,2,3$ yields
\[
a_{1j}^2+a_{2j}^2+a_{3j}^2=1,\quad j=1,2,3,
\]
and for $i\neq j$, $\phi(\xi_i)_{(0)}\phi(\xi_j)=-\frac{1}{2}\widehat{\partial}\phi(\xi_i\xi_j)$ implies
\[
a_{1i}a_{1j}+a_{2i}a_{2j}+a_{3i}a_{3j}=0.
\]
Thus, $A=(a_{ij})\in\Oth_3(D)$.

Finally, we consider $r_j,s_j,j=1,2,3$. Since $\phi(1)_{(1)}\phi(\xi_j)=-\frac{3}{2}\phi(\xi_j),\break j=1,2,3$,
\[
\frac{1}{2}\epsilon_{lmn}(a_{mj}r_n-a_{nj}r_m)=\delta(a_{lj}),\quad\text{and}\quad r_1a_{1j}+r_2a_{2j}+r_3a_{3j}=s_j,
\]
for $l,j=1,2,3$. Writing the first equation in matrix form, we obtain
\begin{align}
&\frac{1}{2}\begin{pmatrix}0&r_3&-r_2\\-r_3&0&r_1\\r_2&-r_1&0\end{pmatrix}A=\delta(A)\label{eqN34}\\
&\quad \Longrightarrow\begin{pmatrix}0&r_3&-r_2\\-r_3&0&r_1\\r_2&-r_1&0\end{pmatrix}
=2\delta(A)A^\mathrm{T},\nonumber
\end{align}
because $A\in\Oth_3(D)$. A direct computation shows that
\begin{align*}
2A^\mathrm{T}\delta(A)&=A^\mathrm{T}\begin{pmatrix}0&r_3&-r_2\\-r_3&0&r_1\\r_2&-r_1&0\end{pmatrix}A \\
&= \begin{pmatrix}0&\displaystyle\sum\limits_{l=1}^3r_lA_{l3}&-\displaystyle\sum\limits_{l=1}^3r_lA_{l2}\\
-\displaystyle\sum\limits_{l=1}^3r_lA_{l3}&0&\displaystyle\sum\limits_{l=1}^3r_lA_{l1}\\
\displaystyle\sum\limits_{l=1}^3r_lA_{l2}&-\displaystyle\sum\limits_{l=1}^3r_lA_{l1}&0
\end{pmatrix}.
\end{align*}
Since $A\in\Oth_3(D)$, $A_{ij}=\det(A) a_{ij}, i,j=1,2,3$. So
\[
\sum_{l=1}^3r_lA_{lj}=\det(A)\sum_{l=1}^3r_la_{lj}=\det(A) s_j.
\]
Hence,
\begin{equation}
\begin{pmatrix}0&s_3&-s_2\\-s_3&0&s_1\\s_2&-s_1&0\end{pmatrix}=2\det(A) A^\mathrm{T}\delta(A).\label{eqN35}
\end{equation}
Summarizing $(\ref{eqN31})$--$(\ref{eqN35})$, we obtain $\phi(f)=\phi_A(f)$, for all $f\in\Lambda(3)$. Hence,
$\phi=\phi_A.$
\end{proof}

\begin{corollary}
$\mathbf{GrAut}(\mathcal{K}_N), N=1,2,3$, is representable by a smooth\break affine $\C$-scheme of finite type.\qed
\end{corollary}

Recall that for the classification of twisted loop Lie
conformal superalgebras based on $\mathcal{K}_N$ the
crucial group to compute is
$\mathbf{Aut}(\mathcal{K}_N)(\widehat{\mathcal{S}}).$ The
next theorem shows that $\mathbf{Aut}(\mathcal{K}_N)$ and
$\mathbf{O}_N$ coincide when evaluated at complex
differential rings that are integral domains. In particular
for $\widehat{S}.$
\begin{theorem} \label{thm-auto-loop}
Let $\mathcal{D}=(D,\delta)$ be a complex differential ring such that $D$ is an integral domain. Then for $N=1,2,3$ the natural inclusion
\[
\mathbf{GrAut}(\mathcal{K}_N)(\mathcal{D})  \subseteq \mathbf{Aut}(\mathcal{K}_N)(\mathcal{D})
\]
is an equality.
\end{theorem}
\begin{proof}
We have to show that every $\phi\in\mathbf{Aut}(\mathcal{K}_N)(\mathcal{D})$ satisfies
\[
\phi\left(\C\otimes_\C\Lambda(N)\otimes_\C D\right)\subseteq\C\otimes_\C\Lambda(N)\otimes_\C D.
\]

\textbf{Case $N=1$:} Since $(\mathcal{K}_1\otimes_\C \mathcal{D})_{\odd}=\C[\partial]\otimes_\C\xi_1\otimes_\C D$ and $\phi$ preserve the $\Z/2\Z$-grading of $\mathcal{K}_1\otimes_\C \mathcal{D}$, we may assume that
\[
\phi(\xi_1)=\sum\limits_{n=0}^{M}\widehat{\partial}^n(\xi_1\otimes s_n),
\]
where $s_n\in D,n=0,\ldots,M$ and $s_M\neq0$. Then
\begin{align*}
[\phi(\xi_1)_\lambda\phi(\xi_1)]&=-\frac{1}{2}\sum\limits_{m=0}^{M}\sum\limits_{n=0}^{M}(-\lambda)^m(\widehat{\partial}+\lambda)^n(1\otimes s_ms_n) \,\, \text{\rm while }\\
\phi([{\xi_1}_\lambda\xi_1])&=\phi\left(\left(-\frac{1}{2}\right)1\right)=-\frac{1}{2}\phi(1).
\end{align*}
The leading term (i.e., the term with highest degree with
respect to $\lambda$) in the right-hand sides of above two
equations are $\frac{1}{2}(-1)^{M+1}\lambda^{2M}(1\otimes
s_M^2)$ and $-\frac{1}{2}\phi(1)$, respectively. Since $D$
is an integral domain, $s_M\neq 0$ implies $s_M^2\neq0$.
Thus $[\phi(\xi_1)_\lambda\phi(\xi_1)]=\phi([{\xi_1}_\lambda\xi_1])$
yields $M=0$, i.e.,
\[
\phi(\xi_1)=\xi_1\otimes s_0,\quad 0\neq s_0\in D.
\]
and $\phi(1)=-2[\phi(\xi_1)_\lambda\phi(\xi_1)]=1\otimes s_0^2$. Hence,
\[
\phi\left(\C\otimes\Lambda(1)\otimes D\right)\subseteq\C\otimes\Lambda(1)\otimes D.
\]

\textbf{Case $N=2$:} Write $\phi(\xi_1\xi_2)=\sum\nolimits_{m=0}^M\widehat{\partial}^m(1\otimes s_m)+x$ with $x\in \C[\partial]\otimes\xi_1\xi_2\otimes D$, $s_m\in D, m=0,\ldots,M, s_M\neq0$. Then
\begin{align*}
0&=[\phi(\xi_1\xi_2)_\lambda\phi(\xi_1\xi_2)]\\
&=\sum\limits_{m,n=0}^M(-\lambda)^m(\widehat{\partial}+\lambda)^n(-\partial1\otimes s_ms_n-2\otimes\delta(s_m)s_n-\lambda2\otimes s_ms_n)\\
&\quad+\left[\sum\limits_{m=0}^M\widehat{\partial}^m(1\otimes s_m)_\lambda x\right]+\left[x_\lambda\sum\limits_{m=0}^M\widehat{\partial}^m(1\otimes s_m)\right]+[x_\lambda x].
\end{align*}
Observing that all terms in the last row of the above
equation are contained in $\C[\partial]\otimes\xi_1\xi_2\otimes D$ and
$(\mathcal{K}_2\otimes_\C
D)_{\even}=(\C[\partial]\otimes1\otimes D)\oplus
(\C[\partial]\otimes\xi_1\xi_2\otimes D)$, we obtain
\[
0=\sum\limits_{m,n=0}^M(-\lambda)^m(\widehat{\partial}+\lambda)^n(\partial1\otimes s_ms_n+2\otimes\delta(s_m)s_n+\lambda2\otimes s_ms_n).
\]
By comparing the coefficients of $\lambda$ it follows that
$s_M^2=0$, hence that $s_M=0$ since $D$ is an integral
domain. This contradicts our assumption that $s_M\neq 0$.
Thus
\[
\phi(\xi_1\xi_2)=x=\sum\limits_{m=0}^{M'}\widehat{\partial}^m(\xi_1\xi_2\otimes c_m),
\]
where $c_m\in D, m=1,\ldots,M', c_{M'}\neq0.$

Similarly, we may assume that
\[
\phi(1)=\sum\limits_{m=0}^{\widetilde{M}'}\widehat{\partial}^m(1\otimes r_m')+\sum\limits_{m=0}^{\widetilde{M}}\widehat{\partial}^m(\xi_1\xi_2\otimes r_m),
\]
where $r'_{m'},r_m\in\widehat{\mathcal{S}}, m'=0,\ldots,\widetilde{M}',m=0,\ldots,\widetilde{M}$.
Then $[\phi(1)_\lambda\phi(\xi_1\xi_2)]=\break -(\widehat{\partial}+\lambda)\phi(\xi_1\xi_2)$ yields $r_{\widetilde{M'}}'c_{M'}=0$ if $\widetilde{M'}+M'>0$. Since $c_{M'}\neq0$ and $D$ is an integral domain, $r_{\widetilde{M'}}'=0$ if $\widetilde{M'}+M'>0$. Thus $\widetilde{M'}=M'=0$, i.e.,
\begin{align*}
\phi(\xi_1\xi_2)&=\xi_1\xi_2\otimes c,\\
\phi(1)&=1\otimes r'+\sum\limits_{m=0}^{\widetilde{M}}\widehat{\partial}^m(\xi_1\xi_2\otimes r_m),
\end{align*}
where $0\neq c\in D, 0\neq r'\in D$ and $r_m\in D, m=0,\ldots,\widetilde{M}$.

Now we consider the odd part $(\mathcal{K}_2\otimes_\C\mathcal{D})_{\odd}=\C[\partial]\otimes(\C\xi_1\oplus \C\xi_2)\otimes D$. Write
\[
\phi(\xi_j)=\sum\limits_{m=0}^{M_{1j}}\widehat{\partial}^m(\xi_1\otimes a_{1j,m})+\sum\limits_{n=0}^{M_{2j}}\widehat{\partial}^n(\xi_2\otimes a_{2j,n}),
\]
where $a_{ij,m}\in D, i,j=1,2$, and $ m=0,\ldots,M_{ij}$. Then it follows from\break $[\phi(\xi_1)_\lambda\phi(\xi_1\xi_2)]=-\frac{1}{2}\phi(\xi_2)$ and $[\phi(\xi_2)_\lambda\phi(\xi_1\xi_2)]=\frac{1}{2}\phi(\xi_1)$ that
\[
\phi(\xi_1)=\xi_1\otimes a_{11}+\xi_2\otimes a_{21},\text{ and }
\phi(\xi_2)=\xi_1\otimes a_{12}+\xi_2\otimes a_{22},
\]
where $a_{ij}\in D,i,j=1,2.$

Next, we consider $\phi(1)$. We deduce from $[\phi(1)_\lambda\phi(\xi_i)]=-(\widehat{\partial}+\frac{3}{2}\lambda)\phi(\xi_i),\break i=1,2$ that $\phi(1)=1\otimes r'+\xi_1\xi_2\otimes r_0$, where $r', r_0\in D$. It follows that:
\[
\phi\left(\C\otimes\Lambda(2)\otimes D\right)\subseteq\C\otimes\Lambda(2)\otimes D.
\]

\textbf{Case $N=3$:} Let $\mathcal{B}=\C[\partial]\otimes_\C(\C\xi_1\xi_2\oplus \C\xi_2\xi_3\oplus \C\xi_3\xi_1)$, then
\[
(\mathcal{K}_3\otimes_\C \mathcal{D})_\even=(\C[\partial]\otimes 1\otimes D)\oplus (\mathcal{B}\otimes_\C\mathcal{D}).
\]
We may assume
$\phi(\xi_i\xi_j)=\sum_{m=0}^M\widehat{\partial}^m(1\otimes
s_m)+x_{ij}, i\neq j$, where $s_m\in D$ and
$x_{ij}\in\mathcal{B}\otimes_\C\mathcal{D}$. Since
$\mathcal{B}\otimes_\C \mathcal{D}$ is an ideal of
$(\mathcal{K}_3\otimes_\C \mathcal{D})_\even$ and
$[\phi(\xi_i\xi_j)_\lambda\phi(\xi_i\xi_j)]=0$, we deduce
that $s_M^2=0$. Then $s_M=0$ because $D$ is an integral
domain. It follows that
$\phi(\xi_i\xi_j)=x_{ij}\in\mathcal{B}\otimes_\C
\mathcal{D}$. Hence, $\phi|_{\mathcal{B}\otimes_\C
\mathcal{D}}$ is an automorphism of the $\mathcal{D}$-Lie
conformal superalgebra $\mathcal{B}\otimes_\C \mathcal{D}$.
Since $\mathcal{B}\simeq\mathrm{Curr}(\mathfrak{so}_3(\C))$
and $\mathfrak{so}_3(\C)$ is a finite-dimensional simple
complex Lie algebra, by \cite{KLP} Corollary~3.17,
we~obtain
\begin{equation}
\phi(\xi_i\xi_j)=\epsilon_{ijl}(\xi_1\xi_2\otimes b_{3l}+\xi_2\xi_3\otimes b_{1l}+\xi_3\xi_1\otimes b_{2l}),\quad i\neq j,\label{eqN37}
\end{equation}
where $(b_{l'l})_{3\times 3}\in\mathbf{GL}_3(D)$. Next we consider $\phi(1)$.\vskip6pt

\textit{Claim:}
\begin{equation}
\phi(1)=1+\xi_1\xi_2\otimes r_3+\xi_2\xi_3\otimes r_1+\xi_3\xi_1\otimes r_2,\qquad r_1,r_2,r_3\in D.\label{eqN38}
\end{equation}

Indeed, we can write
$\phi(1)=\sum_{m=0}^M\widehat{\partial}^m(1\otimes
s_m)+x$ with $s_i\in D, i=0,\ldots,\break M, x\in
\mathcal{B}\otimes_{\C}{\mathcal{D}}$. We may assume
$s_M\neq 0$ because $\phi$ is an isomorphism. Then
\begin{align*}
[\phi(1)_\lambda\phi(1)]&=\!\sum_{m,n=0}^M\!(-\lambda)^m(\widehat{\partial}+\lambda)^{n}(-\partial1\otimes s_ms_{n}-2\otimes\delta(s_m)s_{n}-\lambda2\otimes s_ms_{n})\\
&\quad+\left[\sum\limits_{m=0}^M\widehat{\partial}^m(1\otimes s_m)_\lambda x\right]
+\left[x_\lambda\sum\limits_{n=0}^M\widehat{\partial}^n(1\otimes s_n)\right]+[x_\lambda x].
\end{align*}
Note that all terms in the second row of the above equation
are contained in $\C[\lambda]\otimes_{\C}\mathcal{B}\otimes_{\C}\mathcal{D}$.
If $M>0$, we deduce that $s_M^2=0$ by comparing the
coefficients of $\lambda^{2M+1}$ in
$[\phi(1)_\lambda\phi(1)]=-(\widehat{\partial}+2\lambda)\phi(1)$.
Since $D$ is an integral domain, $s_M=0$. This contradicts
$s_M\neq0$. Hence, $M=0$, i.e., $\phi(1)=1\otimes s_0+x$
with $x\in\mathcal{B}\otimes_{\C}\mathcal{D}$.
\begin{align*}
[\phi(1)_\lambda\phi(1)]&=-\partial1\otimes s_0^2-2\otimes\delta(s_0)s_0-\lambda2\otimes s_0^2\\
&\quad +[(1\otimes s_0)_\lambda x]+[x_\lambda(1\otimes s_0)]+[x_\lambda x],\\
-(\widehat{\partial}+2\lambda)\phi(1)&=-\partial 1\otimes s_0-1\otimes\delta(s_0)-\lambda2\otimes s_0-(\widehat{\partial}+2\lambda)x.
\end{align*}
It follows that $s_0^2=s_0$. Since $s_0\neq0$, $s_0=1$
because $D$ is an integral domain, i.e., $\phi(1)=1+x$ with
$x\in\mathcal{B}\otimes_{\C}\mathcal{D}$.

We further write $x=\sum\nolimits_{l=0}^N\widehat{\partial}^l(\xi_1\xi_2\otimes
r_{3l}+\xi_2\xi_3\otimes r_{1l}+\xi_3\xi_1\otimes r_{2l}),
r_{2l}\in D$. Observing that
$\phi(\xi_i\xi_j)=\epsilon_{ijk}(\xi_1\xi_2\otimes
b_{3k}+\xi_2\xi_3\otimes b_{1k}+\xi_3\xi_1\otimes b_{2k})$, we\break obtain
\begin{align*}
[\phi(1)_\lambda\phi(\xi_i\xi_j)]&=-\epsilon_{ijk}(\partial\xi_1\xi_2\otimes b_{3k}+\partial\xi_2\xi_3\otimes b_{1k}+\partial\xi_3\xi_1\otimes b_{2k})  \\
&\quad -\epsilon_{ijk}\lambda(\xi_1\xi_2\otimes b_{3k}+\xi_2\xi_3\otimes b_{1k}+\xi_3\xi_1\otimes b_{2k})\\
&\quad +\epsilon_{ijk}\sum\limits_{l=0}^N(-\lambda)^l(\xi_1\xi_2\otimes(r_{1l}b_{2k}-r_{2l}b_{1k})\\
&\quad +\xi_2\xi_3\otimes(r_{2l}b_{3k}-r_{3l}b_{2k}))\\
&\quad +\epsilon_{ijk}\sum\limits_{l=0}^N(-\lambda)^l\xi_3\xi_1\otimes(r_{3l}b_{1k}-r_{1l}b_{3k}).
\end{align*}

\noindent Note that
\begin{align*}
-(\widehat{\partial}+\lambda)\phi(\xi_i\xi_j)&=-\epsilon_{ijk}(\partial\xi_1\xi_2\otimes b_{3k}+\partial\xi_2\xi_3\otimes b_{1k}+\partial\xi_3\xi_1\otimes b_{2k})\\
&\quad -\epsilon_{ijk}(\xi_1\xi_2\otimes \delta(b_{3k})+\xi_2\xi_3\otimes \delta(b_{1k})+\xi_3\xi_1\otimes \delta(b_{2k}))\\
&\quad -\epsilon_{ijk}\lambda(\xi_1\xi_2\otimes b_{3k}+\xi_2\xi_3\otimes b_{1k}+\xi_3\xi_1\otimes b_{2k}).
\end{align*}
Then $[\phi(1)_\lambda\phi(\xi_i\xi_j)]=-(\widehat{\partial}+\lambda)\phi(\xi_i\xi_j), i,j=1,2,3, i\neq j$ imply that
\[
r_{il}b_{jk}-r_{jl}b_{ik}=0,
\]
for all $i,j,k=1,2,3, i\neq j, l\geqslant1$. In matrix form, these are equivalent to
\[
\begin{pmatrix}0&-r_{3l}&r_{2l}\\ r_{3l}&0&-r_{1l}\\ -r_{2l}&r_{1l}&0\end{pmatrix}
\begin{pmatrix}b_{11}&b_{12}&b_{13}\\b_{21}&b_{22}&b_{23}\\b_{31}&b_{32}&b_{33}\end{pmatrix}=0, \quad\forall l\geqslant1.
\]
Hence, $r_{1l}=r_{2l}=r_{3l}=0$ for all $l\geqslant1$ because $(b_{ij})_{3\times3}\in\mathbf{GL}_3(D)$, i.e.,
\[
\phi(1)=1+\xi_1\xi_2\otimes r_{30}+\xi_2\xi_3\otimes r_{10}+\xi_3\xi_1\otimes r_{20},
\]
where $r_{10},r_{20},r_{30}\in D$. This completes the proof of the claim.

Next, we consider the odd part $(K_3\otimes_\C
\mathcal{D})_{\odd}$. First
$[\phi(\xi_i\xi_j)_\lambda\phi(\xi_1\xi_2\xi_3)]=0$ for all
$i\neq j$ yield that
$\phi(\xi_1\xi_2\xi_3)=\sum\nolimits_{m=0}^M\widehat{\partial}^m(\xi_1\xi_2\xi_3\otimes
c_m), c_m\in D$. Considering
$[\phi(1)_\lambda\phi(\xi_1\xi_2\xi_3)]=-(\widehat{\partial}+\frac{1}{2}\lambda)\phi(\xi_1\xi_2\xi_3)$,
we deduce that
\begin{equation}
\phi(\xi_1\xi_2\xi_3)=\xi_1\xi_2\xi_3\otimes c,\quad 0\neq c\in D.\label{eqN39}
\end{equation}
A similar argument applies to $\phi(\xi_j)$. From
$[\phi(\xi_j)_\lambda\phi(\xi_1\xi_2\xi_3)]=\epsilon_{jmn}\phi(\xi_m\xi_n)$ and
$[\phi(1)_\lambda\phi(\xi_j)]=-(\widehat{\partial}+\frac{3}{2}\lambda)\phi(\xi_j)$, we obtain
\begin{equation}
\phi(\xi_j)=\xi_1\otimes a_{1j}+\xi_2\otimes a_{2j}+\xi_3\otimes a_{3j}+\xi_1\xi_2\xi_3\otimes s_j,\label{eqN310}
\end{equation}
where $a_{ij}, s_j\in D$. Summarizing (\ref{eqN37}) to (\ref{eqN310}), we obtain
\[
\phi\left(\C\otimes\Lambda(3)\otimes D\right)\subseteq\left(\C\otimes\Lambda(3)\otimes D\right).
\]
This completes the proof.
\end{proof}

\begin{remark}
The integral assumption on $D$ is not superflous. Consider
the complex differential ring $\mathcal{D}=(D,\delta)$,
where $D=\C\oplus\C\tau, \tau^2=0$ (the algebra of dual
numbers) and $\delta=0.$ For the $\mathcal{D}$-Lie
conformal superalgebra
$\mathcal{K}_1\otimes_\C\mathcal{D}=\C[\partial]\otimes_\C(\C\oplus\C\xi_1)\otimes_\C
D$ it is easy to check that
\begin{align*}
    \phi(\partial^\ell\otimes1\otimes s)&=\partial^\ell\otimes1\otimes s+\partial^{\ell+1}\otimes1\otimes\tau s, \ell\geqslant0, s\in D,\\
\phi(\partial^\ell\otimes\xi_1\otimes s)&=\partial^\ell\otimes\xi_1\otimes s+\partial^{\ell+1}\otimes\xi_1\otimes\tau{s}, \ell\geqslant0, s\in D,
\end{align*}
define an automorphism of the $\mathcal{D}$-Lie conformal
superalgebra of $\mathcal{K}_1\otimes_\C\mathcal{D}$. But
$\phi\not\in\mathbf{GrAut}(\mathcal{K}_1)(\mathcal{D})$.
\end{remark}

\section{Forms of the $N=1,2,3$ Lie conformal superalgebras}

In this section, we classify the twisted loop Lie conformal
superalgebras based on the complex Lie conformal
superalgebra $\mathcal{K}_N, N=1,2,3$. As previously
mentioned, the classification can be completed in two
steps. The first step is to classify the
$\widehat{\mathcal{S}}/\mathcal{R}$-forms of
$\mathcal{K}_N\otimes_\C \mathcal{R}$, and  then look at
the passage from isomorphic classes of the
$\mathcal{R}$-Lie conformal superalgebras to isomorphic
classes of the complex Lie conformal superalgebras. Recall
that $\widehat{S}=\C[t^q, q\in\Q]$ is the algebraic simply
connected cover of $R = \C[t^{\pm 1}]$, and that the
algebraic fundamental group $\pi_1(R)$ of
$\mathrm{Spec}(R)$ at the geometric point
$\mathrm{Spec}\big(\overline{\C(t)}\big)$ can be identified
with $\widehat{\Z}=\lim\limits_{\longleftarrow}\Z/m\Z$ via
our canonical choice of compatible roots of unity $\zeta_m
= e^{\frac{\mathbf{i}2\pi}{m}}.$ The explicit continuous
action of the profinite group  $\widehat{\Z}$  on
$\widehat{S}$ is given by
$^{\overline{1}}t^{p/m}=\zeta_m^pt^{p/m}$, where
$\overline{1}$ is the image of $1$ in $\widehat{\Z}$ under
the canonical homomorphism
$\Z\rightarrow\widehat{\Z}$.\footnote{The continuous action
of $\widehat{\Z}$ is determined by the action of $\bar{1}$
because $\Z$ is dense in $\widehat{\Z}$.} We thus have
natural {\it continuous} actions of $\pi_1(R)$ on
$\mathbf{Aut}(\mathcal{K}_N)(\widehat{\mathcal{S}})$ and
$\mathbf{O}_N(\widehat{S}), N=1,2,3$. Note that the
isomorphism $\iota:\mathbf{O}_N(\widehat{S})\rightarrow\mathbf{Aut}(\mathcal{K}_N)(\widehat{\mathcal{S}})$
is  $\pi_1(R)$-equivariant.

\begin{theorem}\label{RSform}
Let $N=1,2,3.$ There are exactly two
$\widehat{\mathcal{S}}/\mathcal{R}$-forms (up to
isomorphism of $\mathcal{R}$-Lie conformal superalgebras)
of $\mathcal{K}_N\otimes_{\C}\mathcal{R}.$ These are
$\mathcal{L}(\mathcal{K}_N,\mathrm{id})$ and
$\mathcal{L}(\mathcal{K}_N,\omega_N)$, where
$\omega_N:\mathcal{K}_N\rightarrow\mathcal{K}_N$ is the
automorphism of the complex Lie conformal superalgebra
$\mathcal{K}_N$ given by
\begin{align*}
&\omega_1:&&1\mapsto1, &&\xi_1\mapsto-\xi_1,\\
&\omega_2:&&1\mapsto1, &&\xi_1\mapsto-\xi_1,\\ &&&\xi_2\mapsto\xi_2,&&\xi_1\xi_2\mapsto-\xi_1\xi_2,\\
&\omega_3:&&1\mapsto1,&&\xi_j\mapsto-\xi_j,j=1,2,3,\\
&&&\xi_i\xi_j\mapsto\xi_i\xi_j,i\neq j, &&\xi_1\xi_2\xi_3\mapsto-\xi_1\xi_2\xi_3.
\end{align*}
\end{theorem}
\begin{proof}
By \cite{KLP} Theorem~2.16,  the
$\widehat{\mathcal{S}}/\mathcal{R}$-forms of
$\mathcal{K}_N\otimes_\C \mathcal{R}$ are parametrized by
the non-abelian \v{C}ech cohomology set
$H^1\big(\widehat{\mathcal{S}}/\mathcal{R},\mathbf{Aut}(\mathcal{K}_N)\big)$.

Since $\mathcal{K}_N\otimes_{\C}\mathcal{R}$ is spanned by
\[
\{\widehat{\partial}^\ell(rf)|r\in R, \ell\geqslant0, f=\xi_{i_1}\ldots\xi_{i_t}, 1\leqslant i_1<\cdots<i_t\leqslant N\},
\]
we see that $\mathcal{R}$-Lie conformal superalgebra
$\mathcal{K}_N\otimes_{\C}\mathcal{R}$ satisfies the
finiteness condition of Proposition~2.29 of \cite{KLP}, and
this allows us to identify the cohomology set
$H^1\big(\widehat{\mathcal{S}}/\mathcal{R},\mathbf{Aut}(\mathcal{K}_N)\big)$
with the ``usual'' non-abelian (continuous) cohomology set
$H^1\big(\pi_1(R),\mathbf{Aut}(\mathcal{K}_N)(\widehat{\mathcal{S}})\big).$
Our problem is thus reduced to computing
$H^1(\pi_1(R),\mathbf{Aut}(\mathcal{K}_N)(\widehat{\mathcal{S}}))$.

The loop algebras $\mathcal{L}(\mathcal{K}_N,\mathrm{id})$
and $\mathcal{L}(\mathcal{K}_N,\omega_N)$  correspond to
the classes $[\alpha]$ and $[\beta]$ in
$H^1\big(\pi_1(R),\mathbf{Aut}(\mathcal{K}_N)(\widehat{\mathcal{S}})\big)$
given by the (constant) cocycles
$\alpha,\beta:\pi_1(R)\rightarrow\mathbf{Aut}(\mathcal{K}_N)(\widehat{\mathcal{S}})$
defined by $1\mapsto \mathrm{id}$ and $ 1\mapsto\omega_N$,
respectively.

Since the isomorphism $\iota:\mathbf{O}_N(\widehat{S})\rightarrow\mathbf{Aut}(\mathcal{K}_N)(\widehat{\mathcal{S}})$
of Theorem~\ref{auto-thm} is\break $\pi_1(R)$-equivariant we are
reduced to determining $H^1\big(\widehat{\Z},\mathbf{O}_N(\widehat{S})\big).$ The
classes in $H^1\big(\widehat{\Z},\mathbf{O}_N(\widehat{S})\big)$
corresponding to $[\alpha]$ and $[\beta]$ will be still denoted in the same way.

Consider the split exact sequence of $\widehat{\Z}$-groups
\begin{equation}
1\longrightarrow\mathbf{SO}_N(\widehat{S})\longrightarrow\mathbf{O}_N(\widehat{S})\overset{\det}{\longrightarrow}\Z/2\Z\rightarrow1,\label{eqCoh1}
\end{equation}
where $\widehat{\Z}$ acts on $\Z/2\Z$ trivially and
``$\det$'' is the determinant map. It yields the exact
sequence of  pointed sets
\begin{equation}
H^1\big(\widehat{\Z},\mathbf{SO}_N(\widehat{S})\big)\longrightarrow H^1\big(\widehat{\Z},\mathbf{O}_N(\widehat{S})\big)\overset{\psi}{\longrightarrow}
H^1(\widehat{\Z},\Z/2\Z).\label{eqCoh2}
\end{equation}
Since $\widehat{\Z}$ acts on $\Z/2\Z$ trivially, we have
$H^1(\widehat{\Z},\Z/2\Z)\simeq \Z/2\Z$. Since the short
exact sequence $(\ref{eqCoh1})$ is split, $\psi$ admits a
section hence is surjective by general considerations. This
is also explicitly clear in our situation since  $\psi$
visibly maps $[\alpha]$ and $[\beta]$ to the two distinct
classes in $H^1(\widehat{\Z},\Z/2\Z)$. It remains to show
that $\psi$ is injective.

The fiber of $\psi$ over the trivial class of
$H^1(\widehat{\Z},\Z/2\Z)$ is measured by
$H^1\big(\widehat{\Z},\mathbf{SO}_N(\widehat{S})\big)$,
while the fiber over the non-trivial class is measured by
$H^1\big(\widehat{\Z},{_\beta}\mathbf{SO}_N(\widehat{S})\big)$
where ${_\beta}\mathbf{SO}_N$ is the group scheme over $R$
obtained from $\mathbf{SO}_N$ by twisting  by
$\beta.$\footnote{Strictly speaking the twist is by the
pullback $\psi_{*}([\beta]).$  The abuse of terminology and
notation is standard.} Since every finite connected
\'{e}tale cover of\break $\mathrm{Spec}(\C[t^{\pm1}])$ is
of the form $\mathrm{Spec}(\C[t^{\pm\frac{1}{m}}])$ with
$m$ a positive integer, all such  covers have  trivial
Picard group. By \cite{P1}, Theorem~3.1 (i)
$H^1_{\mathrm{\acute{e}t}}(R,\mathfrak{G})$ vanishes for
every reductive group scheme $\mathfrak{G}$ over  $R$, in
particular for $\mathbf{SO}_N$ and ${_\beta}\mathbf{SO}_N.$
On the other hand by the isotriviality theorem of
\cite{GP1,GP2} we have
$H^1_{\mathrm{\acute{e}t}}(R,\mathfrak{G}) \simeq
H^1\big(\pi_1(R),\mathfrak{G}(\widehat{S})\big)=
H^1\big(\widehat{\Z},\mathfrak{G}(\widehat{S})\big).$ This
finishes the proof of injectivity.
\end{proof}

Now we consider the relation between isomorphism of twisted
loop algebras based on $\mathcal{K}_N$ as $\mathcal{R}$-Lie
conformal superalgebras and isomorphism of these objects as
complex Lie conformal superalgebras. Recall that for any
$\mathcal{R}$-Lie conformal superalgebra $\mathcal{A}$, the
centroid\footnote{The centroid of a $\mathcal{R}$-Lie
conformal superalgebra is originally defined in \cite{KLP}
Section~2.4, which is analogous to the centroid of root
graded Lie algebras studied in \cite{BN}.}
$\mathrm{Ctd}_{\mathcal{R}}(\mathcal{A})$ is defined to be
\[
\mathrm{Ctd}_{\mathcal{R}}(\mathcal{A})=\{\chi\in\mathrm{End}_{R\text{-smod}}(\mathcal{A})|\chi(a_{(n)}b)=a_{(n)}\chi(b)\text{ for all }a,b\in\mathcal{A},n\in\Bbb N\},
\]
where $\mathrm{End}_{R\text{-smod}}(\mathcal{A})$ is the
set of homogeneous $R$-supermodule endomorphisms
$\mathcal{A}\rightarrow\mathcal{A}$ of degree $\even$.

\begin{lemma}\label{ctd}
Let $\mathcal{A}=\mathcal{L}(\mathcal{K}_N,\sigma)$ where
$\sigma$ is an automorphism of $\mathcal{K}_N$ of order
$m$, where $N=1,2,3$. Then the canonical map
$R\rightarrow\mathrm{Ctd}_\C(\mathcal{A})$ is a
$\C$-algebra isomorphism.
\end{lemma}
\begin{proof}
Recall that for $r\in R$, there is an endomorphism
$r_{\mathcal{A}}:\mathcal{A}\rightarrow\mathcal{A},
a\mapsto ra$, which is an element of $\mathrm{Ctd}_\Bbb
C(\mathcal{A})$. Then the canonical map
$R\rightarrow\mathrm{Ctd}_{\Bbb C}(\mathcal{A})$ is defined
by $r\mapsto r_{\mathcal{A}}$. Since
$\sigma\circ\partial=\partial\circ\sigma$, we obtain,
\begin{align}
(\mathcal{K}_N)_i&=\{a\in\mathcal{K}_N|\sigma(a)=\zeta_m^ia\}\label{eigen}\\
&=\C[\partial]\otimes_\C\mathrm{span}_\C \{f\in\Lambda(N)|\sigma(f)=\zeta_m^if\text{ and }f\text{ homogeneous}\}.\notag
\end{align}
Recall that
\[
\mathcal{A}=\sum\limits_{i=0}^{m-1}(\mathcal{K}_N)_i\otimes t^{\frac{i}{m}}{\Bbb C}[t^{\pm1}]\subseteq\mathcal{K}_N\otimes_\C \widehat{\mathcal{S}},
\]
and
\[
\mathcal{K}_N\otimes_\C\widehat{\mathcal{S}}=\bigoplus\limits_{l\geqslant0\atop k=0,\ldots,N}(\mathcal{K}_N\otimes_\C\widehat{\mathcal{S}})_{l,k},
\]
where $(\kern-1pt \mathcal{K}_N\otimes_{\Bbb C}\widehat{\mathcal{S}})_{l,k}=\{\!(\kern-1pt \partial^lf\kern-1pt )\otimes
s|f\in\Lambda(\kern-1.5pt N\kern-1pt )\text{ homogeneous of degree }k, s\in\widehat{S}\}$. Hence,
\begin{align}
\mathcal{A}&=\bigoplus\limits_{l\geqslant0\atop k=0,\ldots,N}\left(\mathcal{A}\cap (\mathcal{K}_N\otimes_{\Bbb C}\widehat{\mathcal{S}})_{l,k}\right) \label{span}\\
&=\mathrm{span}_\C\left\{(\partial^lf)\otimes s\middle|f\in\Lambda(N)\text{ homogeneous, }\atop s\in\widehat{S}\text{ and }\ell\geqslant0\text{ such that }\partial^lf\otimes s\in\mathcal{A}\right\}.\notag
\end{align}

Let $\chi\in\mathrm{Ctd}_\C(\mathcal{A})$. Let $L=-1\otimes1\in\mathcal{A}$. For $f\in\Lambda(N)$ homogeneous of degree $|f|$,
\[
L_{(1)}((\partial^\ell f)\otimes s)=\left(2+\ell-\frac{1}{2}|f|\right)(\partial^\ell f)\otimes s, \qquad\forall s\in\widehat{S}, \ell\geqslant0.
\]
We therefore see that $\mathcal{K}_N\otimes_\C \widehat{\mathcal{S}}$ can be decomposed into the direct sum of eigenspaces of $L_{(1)}$ and
\begin{align*}
(\mathcal{K}_N\otimes_\C\widehat{\mathcal{S}})_2&=\{a\in \mathcal{K}_N\otimes_\C\widehat{\mathcal{S}}|L_{(1)}a=2a\}\\
&=\mathrm{span}_\C \{1\otimes s,(\partial\xi_i\xi_j)\otimes s'|i<j, s,s'\in\widehat{S}\}.
\end{align*}
Applying $\chi$ on $L_{(1)}L=2L$ we obtain $L_{(1)}\chi(L)=2\chi(L)$, thus we may assume that
\[
\chi(L)=L\otimes r+\sum\limits_{i<j}(\partial\xi_i\xi_j)\otimes s_{ij},
\]
where $r,s_{ij}\in\widehat{S}$. Then
\[
0=\chi(L_{(2)}L)=L_{(2)}(L\otimes s)+\sum\limits_{i<j}L_{(2)}\big((\partial\xi_i\xi_j)\otimes s_{ij}\big)=2\sum\limits_{i<j}\xi_i\xi_j\otimes s_{ij}.
\]
It follows $s_{ij}=0, i<j$. So $\chi(L)=L\otimes r$ and $r\in R$ because $\chi(L)=L\otimes\break r\in\mathcal{A}$.

For $f\in\Lambda(N)$ homogeneous of degree $|f|$ and $s\in\widehat{S}$ such that $f\otimes s\in\mathcal{A}$,
\begin{align*}
\left(2-\frac{1}{2}|f|\right)\chi(f\otimes s)&=\chi((f\otimes s)_{(1)}L)=(f\otimes s)_{(1)}\chi(L)\\
&=(f\otimes s)_{(1)}L\otimes r=\left(2-\frac{1}{2}|f|\right)f\otimes sr.
\end{align*}
Since $|f|\leqslant3$, $2-\frac{1}{2}|f|\neq0$, thus $\chi(f\otimes s)=f\otimes sr$.

Let $f\in\Lambda(N)$ homogeneous of degree $|f|$,
$s\in\widehat{S}$ and $\ell\geqslant0$ such that
$(\partial^{l+1})f\otimes s\in\mathcal{A}$. By
$(\ref{eigen})$, we have $(\partial^lf)\otimes
s\in\mathcal{A}$. Then
\[
\chi\big((\partial^{l+1}f)\otimes s\big)=\chi\big(L_{(0)}((\partial^lf)\otimes s)\big)=L_{(0)}\chi\big((\partial^lf)\otimes s\big).
\]
By induction, $\chi\big((\partial^\ell f)\otimes s\big)=(\partial^\ell f)\otimes sr$ for all $\ell\geq0, s\in\widehat{S}$ such that\break $(\partial^\ell f)\otimes s\in\mathcal{A}$.

Therefore, $(\ref{span})$ implies $\chi(a)=ar$ for all $a\in\mathcal{A}$. Thus the canonical map $R\rightarrow\mathrm{Ctd}_\C (\mathcal{A})$ is surjective. Injectivity is clear.
\end{proof}

\begin{theorem}
There are exactly two twisted loop Lie conformal
superalgebra (up to isomorphism of complex Lie conformal
superalgebras) based on each $\mathcal{K}_N, N=1,2,3$.
These are $\mathcal{L}(\mathcal{K}_N,\mathrm{id})$ and
$\mathcal{L}(\mathcal{K}_N,\omega_N)$.
\end{theorem}
\begin{proof}
Each twisted loop Lie conformal superalgebra based on
$\mathcal{K}_N$ is an
$\widehat{\mathcal{S}}/\mathcal{R}$-form of
$\mathcal{K}_N\otimes_\C \mathcal{R}$. It follows from
Theorem~\ref{RSform} that there exists exactly two of them
up to isomorphism of $\mathcal{R}$-Lie conformal
superalgebras, namely
$\mathcal{L}(\mathcal{K}_N,\mathrm{id})$ and
$\mathcal{L}(\mathcal{K}_N,\omega_N)$. By Lemma~\ref{ctd}
and Corollary~2.36 of  \cite{KLP} we conclude that
$\mathcal{L}(\mathcal{K}_N,\mathrm{id})$ and
$\mathcal{L}(\mathcal{K}_N,\omega_N)$ remain non-isomorphic
when viewed as complex Lie conformal
superalgebras.\footnote{The situation is quite different
for loop algebras based on finite-dimensional simple
algebras. Here non-isomorphic objects over $R$ may become
isomorphic when viewed as complex algebras. Remarkably
enough, this does not happen in the case of Lie algebras.
The reason is that in the outer group of automorphisms
(symmetries of the Dynkin diagram), every element is
conjugate to its inverse. See \cite{P1} for details.}
\end{proof}

\section{Passage from Lie conformal superalgebras to superconformal Lie algebras}

In the previous section, we have shown that there are only
two twisted loop Lie conformal superalgebras based on
$\mathcal{K}_N$ up to isomorphism of complex Lie conformal
superalgebras, namely
$\mathcal{L}(\mathcal{K}_N,\mathrm{id})$ and
$\mathcal{L}(\mathcal{K}_N,\omega_N)$. By factoring the
image of $\partial+\delta_t$, we obtain two Lie
superalgebras $\mathrm{Alg}(\mathcal{K}_N,\mathrm{id})$ and
$\mathrm{Alg}(\mathcal{K}_N,\omega_N)$. The crucial point
is that these two Lie superalgebras are non-isomorphic.
They can be distinguished by the eigenvalues of the
Virasoro operator.\footnote{For $N = 4$ this argument is
mentioned, without proof, in \cite{KLP}.} In this section,
we will give a rigorous proof of this non-isomorphism
statement when $N=3$,  which is the most involved of all
three cases.

Let $\mathrm{Alg}(\mathcal{K}_3,\mathrm{id})=\mathcal{L}(\mathcal{K}_3,\mathrm{id})/(\partial+\delta_t)\mathcal{L}(\mathcal{K}_3,\mathrm{id})$,
where the Lie bracket on
$\mathrm{Alg}(\mathcal{K}_3,\mathrm{id})$ is given by the
zeroth product of $\mathcal{L}(\mathcal{K}_3,\mathrm{id})$.
For $a\in\mathcal{L}(\mathcal{K}_3,\mathrm{id})$, let
$\bar{a}$ be its image in
$\mathrm{Alg}(\mathcal{K}_3,\mathrm{id})$. Set
\begin{align*}
&L_m=-\overline{1\otimes t^{m+1}},&& G^i_\alpha=2\overline{\xi_i\otimes t^{\alpha+\frac{1}{2}}},\\
&T^i_m=2\mathbf{i}\epsilon_{ijl}\overline{\xi_j\xi_l\otimes t^m},&&\Psi_\alpha=-2\mathbf{i}\overline{\xi_1\xi_2\xi_3\otimes t^{\alpha-\frac{1}{2}}}.
\end{align*}
for $i=1,2,3, m\in\Z, \alpha\in\frac{1}{2}+\Z$. Then $\{L_m,T^i_m,G^i_\alpha,\Psi_\alpha|i=1,2,3, m\in\Z,\break \alpha\in\frac{1}{2}+\Z\}$ is a basis of $\mathrm{Alg}(K_3,\mathrm{id}).$ For $m,n\in\Z,\alpha,\beta\in\frac{1}{2}+\Z, i,j=1,2,3$, the following relations hold:
\[
\begin{array}{@{}l@{\;\,}lll@{}}
[L_m,L_n]=(m-n)L_{m+n},&[L_m,T^i_n]=-nT^i_{m+n},& [T^i_m,T^j_n]=\mathbf{i}\epsilon_{ijl}T^l_{m+n},\\
\/[L_m,\Psi_{\alpha}]=-\left(\dfrac{1}{2}m+\alpha\right)\Psi_{m+\alpha},&[T^i_m,\Psi_\alpha]=0,& [\Psi_\alpha,\Psi_\beta]=0,\\
\/[L_m,G^i_{\alpha}]=\left(\dfrac{1}{2}m-\alpha\right)G^i_{m+\alpha},&\multicolumn{2}{@{}l}{[T^i_m,G^j_\alpha]=\mathbf{i}\epsilon_{ijl}G^l_{m+\alpha}+\delta_{ij}m\Psi_{m+\alpha},}\\
\/[G^i_{\alpha},\Psi_{\beta}]=T^i_{\alpha+\beta},&\multicolumn{2}{@{}l}{[G^i_\alpha,G^j_\beta]=2\delta_{ij}L_{\alpha+\beta}+\mathbf{i}\epsilon_{ijl}(\alpha-\beta)T^l_{\alpha+\beta}.}\\
\end{array}
\]

Next we consider $\mathrm{Alg}(\mathcal{K}_3,\omega_3)
=\mathcal{L}(\mathcal{K}_3,\omega_3)/(\partial+\delta_t)\mathcal{L}(\mathcal{K}_3,\omega_3)$,
where the Lie bracket  is given by the corresponding zeroth
product  as before. Since  $\omega_3$ acts on  the even
part $(\mathcal{K}_3)_\even$ as the identity and on  the
odd part $(\mathcal{K}_3)_\odd$ as $-\mathrm{id}$, we have
by definition
\[
\mathcal{L}(\mathcal{K}_3,\omega_3)=((\mathcal{K}_3)_{\even}\otimes_\C \C[t,t^{-1}])\oplus ((\mathcal{K}_3)_{\odd}\otimes_{\C}t^{\frac{1}{2}}\C[t,t^{-1}]).
\]
As before for $a\in\mathcal{L}(\mathcal{K}_3,\omega_3)$ we
let $\bar{a}$  denote its image in
$\mathrm{Alg}(\mathcal{K}_3,\omega_3)$. Let
\begin{align*}
&L_m=-\overline{1\otimes t^{m+1}},&& G^i_m=2\overline{\xi_i\otimes t^{m+\frac{1}{2}}},\\
&T^i_m=2\mathbf{i}\epsilon_{ijl}\overline{\xi_j\xi_l\otimes t^m},&&\Psi_m=-2\mathbf{i}\overline{\xi_1\xi_2\xi_3\otimes t^{m-\frac{1}{2}}},
\end{align*}
for $i=1,2,3, m\in\Z$. Then $\{L_m,T^i_m,G^i_m,\Psi_m|i=1,2,3, m\in\Z\}$ is a basis of
$\mathrm{Alg}(\mathcal{K}_3,\omega_3)$ and these elements
satisfy the following relations  for $m,n\in\Z$ and
$i,j=1,2,3$:
\[
\begin{array}{@{}l@{\;\,}lll@{}}
[L_m,L_n]=(m-n)L_{m+n},&[L_m,T^i_n]=-nT^i_{m+n},&[T^i_m,T^j_n]=\mathbf{i}\epsilon_{ijl}T^l_{m+n},\\
\/[L_m,\Psi_n]=-\left(\dfrac{1}{2}m+n\right)\Psi_{m+n},&[T^i_m,\Psi_n]=0,&[\Psi_m,\Psi_n]=0,\\
\/[L_m,G^i_n]=\left(\dfrac{1}{2}m-n\right)G^i_{m+n},&\multicolumn{2}{@{}l}{[T^i_m,G^j_n]=\mathbf{i}\epsilon_{ijl}G^l_{m+n}+\delta_{ij}m\Psi_{m+n},}\\
\/[G^i_m,\Psi_n]=T^i_{m+n},&\multicolumn{2}{@{}l}{[G^i_m,G^j_n]=2\delta_{ij}L_{m+n}+\mathbf{i}\epsilon_{ijl}(m-n)T^l_{m+n}.}\\
\end{array}
\]

To prove that these two Lie superalgebras are not isomorphic, we first establish an auxiliary useful Lemma.
Consider the subspace of $\mathrm{Alg}(\mathcal{K}_3,{\rm id})$ and $\mathrm{Alg}(\mathcal{K}_3,\omega_3)$ with bases composed of  the respective
$\{L_m,T^i_m|i=1,2,3,\break m\in\Z\}.$ The above relations show that these are in fact sub Lie superalgebras which are isomorphic. We will denote them by $\mathfrak{g}.$
\begin{lemma}[{[Rigidity of the Virasoro element]}]\label{auto-even}
If $\phi$ is an automorphism of $\mathfrak{g}$, then $\phi(L_0)=\pm L_0$.
\end{lemma}
\begin{proof}
$\mathfrak{g}$ can be decomposed into the direct sum of weight spaces with respect to $L_0$, i.e.,
\[
\mathfrak{g}=\bigoplus_{n\in\Z}\mathfrak{g}_n,
\]
where $\mathfrak{g}_n=\{x\in\mathfrak{g}|[L_0,x]=nx\}=\mathrm{span}_\C\{L_{-n}, T_{-n}^i,i=1,2,3\}$, $n\in\Z$.  If\break $x\in\mathfrak{g}_n$ we say that $x$ is homogeneous of weight $n$ and write $\mathrm{wt}(x)=n.$ Note that $[\mathfrak{g}_m,\mathfrak{g}_n]\subseteq\mathfrak{g}_{m+n}$.

Recall that an element $x$ of a Lie superalgebra
$\mathfrak{g}$ over $\C$ is called {\it locally finite} if
for every fixed $y\in\mathfrak{g}$, $\{(\mathrm{ad} x)^ny,
n\in \Bbb N\}$ spans a finite-dimensional subspace, or
equivalently, every element of $\mathfrak{g}$ lies inside a
finite-dimensional $\mathrm{ad} x$-stable subspace. We
first show that the locally finite elements of
$\mathfrak{g}$ are precisely the elements of
$\mathfrak{g}_0$. For $n>0$, let $x$ be an arbitrary
non-zero element in $\mathfrak{g}_n$. We claim that there
exists a homogeneous element  $y\in\mathfrak{g}$ such that
$(\mathrm{ad}x)^m(y)\neq0$ and
$\mathrm{wt}\big((\mathrm{ad}x)^m(y)\big)<\mathrm{wt}\big((\mathrm{ad}x)^{m+1}(y)\big)$
for all $m\in \mathbb{N}$.

We may assume $x=aL_{-n}+b_1T^1_{-n}+b_2T^2_{-n}+b_3T^3_{-n}$. If $a\neq0$, let $y=L_{-2n}$, then
\begin{align*}
(\mathrm{ad} x)^m(y)&=m!a^mn^mL_{-(m+2)n}+c_{1m}T^1_{-(m+2)n}+c_{2m}T^2_{-(m+2)n}\\
&\quad +c_{3m}T^3_{-(m+2)n} \\
&\neq 0
\end{align*}
and $\mathrm{wt}\big((\mathrm{ad}x)^my\big)=(m+2)n,
m\in\Bbb N$. If $a=0$,
$x=b_1T^1_{-n}+b_2T^2_{-n}+b_3T^3_{-n}$ with
$(b_1,b_2,b_3)\neq(0,0,0)$. Take $y=L_{-1}$, then
\[
(\mathrm{ad} x)^m(y)=n(n+1)\cdots(n+m-1)(b_1T^1_{-n-m}+b_2T^2_{-n-m}+b_3T^3_{-n-m})\neq0
\]
and $\mathrm{wt}\big((\mathrm{ad}x)^my\big)=m+n$.

Now, let $x\in\mathfrak{g}$ be locally finite, we can
write, $x=x_{-M'}+\cdots+\break x_0+\cdots+x_{M}$ with
$x_l\in\mathfrak{g}_l, l=-M',\ldots,M$ and
$M',M\geqslant0$. If $M>0$ and $x_M\neq0$, then there is a
homogeneous element $y\in\mathfrak{g}$ such that
$\mathrm{wt}\big((\mathrm{ad}
x_M)^my\big)<\mathrm{wt}\big((\mathrm{ad}
x_M)^{m+1}y\big)$. Then
\[
(\mathrm{ad} x)^m(y)=(\mathrm{ad} x_M)^m(y)+\sum\limits_{\mathrm{wt}{z}<\mathrm{wt}((\mathrm{ad} x_M)^m(y))}z.
\]
Thus $\{(\mathrm{ad} x)^m(y)|m\in\Bbb N\}$ is linear
independent. This contradicts the locally finiteness of
$x$. So $M=0$. Similarly, we can prove $M'=0$. Hence,
$x\in\mathfrak{g}_0$.

Conversely, it is easy to check every elements in
$\mathfrak{g}_0$ is locally finite. Moreover, since
$[L_0,L_0]=0,[L_0,T^i_0]=0,
[T^i_0,T^j_0]=\mathbf{i}\epsilon_{ijl}T^l_0$,
$\mathfrak{g}_0$ is a subalgebra of $\mathfrak{g}$,
isomorphic to $\mathfrak{gl}_2(\C)$.

If $\phi$ is an automorphism of $\mathfrak{g}$, then it
preserves locally finite elements. So the restriction of
$\phi$ to $\mathfrak{g}_0$ induces an automorphism of
$\mathfrak{g}_0$. Observe that the center of
$\mathfrak{g}_0$ is $\C L_0$. So $\phi(L_0)=aL_0$ for some
$0\neq a\in \C$.

Since $a[L_0,\phi(L_{-n})]=[\phi(L_0),\phi(L_{-n})]=n\phi(L_{-n})$,
$\phi(L_{-n})$ is a weight vector of $L_0$ with weight
$\frac{n}{a}$ and all weights of $L_0$ are  integers, it
follows $\frac{n}{a}$ is an integer for all $n$. Hence
$a=\pm1$.
\end{proof}

\begin{remark}
Similar results hold for $N = 1,2.$
\end{remark}

\begin{proposition}
For $N = 1,2,3$ the Lie superalgebras $\mathrm{Alg}(\mathcal{K}_N,\mathrm{id})$ and $\mathrm{Alg}(\mathcal{K}_N,\omega_N)$ are not isomorphic.
\end{proposition}
\begin{proof}
As mentioned above we only give an explicit prove in the
case $N = 3.$ Suppose  that
$\phi:\mathrm{Alg}(\mathcal{K}_3,\mathrm{id})\rightarrow
\mathrm{Alg}(\mathcal{K}_3,\omega_3)$ is an isomorphism of
Lie superalgebras. Then it induces an  isomorphism on the
even parts which by rigidity satisfies $\phi(L_0)=\pm L_0$.
But the eigenvalues of $L_0$ in
$\mathrm{Alg}(\mathcal{K}_3,\mathrm{id})_\odd$ are
contained in $\frac{1}{2}+\Z$, while the eigenvalues of
$\phi(L_0)=\pm L_0$ in
$\mathrm{Alg}(\mathcal{K}_3,\omega_3)_\odd$ are all
integers; a contradiction.
\end{proof}

\section{A conjecture on representability}

Let $\mathcal{A}$ be one of the complex $N$-Lie conformal
superalgebras, where $N = 1,2,3,4.$ (In particular,
$\mathcal{A} = \mathcal{K}_N$ if $N \neq 4.$ If $N = 4$ the
description of $\mathcal{A}$ is given in detail in both
\cite{K, KLP}).  Let $\mathcal{L}$ be an
$\widehat{\mathcal{S}}/\mathcal{R}$-form of $\mathcal{A}
\otimes_\C \mathcal{R}$, Consider the corresponding
$\mathcal{R}$-group functor of automorphisms
$\mathbf{Aut}(\mathcal{L})$ (see Remark \ref{grpfunrmk}).
The above construction for $N \neq 4$ and that of
\cite{KLP} in the case $N = 4$ show that there is a
suitable degree 0 subspace $\mathcal{V}$ of $\mathcal{L}$
such that $\mathcal{L}=\bigoplus_{n\in\Bbb
N}\widehat{\partial}^n\mathcal{V}$ since the descent data
defining $\mathcal {L}$ preserves the suitable degree $0$
space $\C \otimes V \otimes \widehat{\mathcal{S}}$ of
$\mathcal{A} \otimes \widehat{\mathcal{S}} = \C[\partial]
\otimes_\C V \otimes_\C \widehat{\mathcal{S}}.$ One can
define in a natural way a subgroup functor
$\mathbf{GrAut}(\mathcal{L})$ of
$\mathbf{Aut}(\mathcal{L}).$

\begin{conjecture}
{\it If $ N \neq 4$ the $\mathcal{R}$-group functor
$\mathbf{GrAut}(\mathcal{L})$ is representable by  an a
affine group scheme of finite type whose connected
component of the identity is simple (in the sense of {\rm
\cite{SGA3}}).}
\end{conjecture}

After a ``differential'' version of faithfully flat descent
we may assume that $\mathcal{L}$ is split, that is
$\mathcal{L} = \mathcal{A} \otimes_\C \mathcal{R}.$ For $N
\neq 4$ we have $\mathbf{GrAut}(\mathcal{L}) =
\mathbf{O}_{N,R} $, where this last denotes  the orthogonal
group $\mathbf{O}_N$ over ${\rm Spec}(R).$

In view of \cite{KLP}, for $N = 4$ the natural conjecture
to make is that in the split case
$\mathbf{GrAut}(\mathcal{L}) = \big(\mathbf{SL}_2 \times
{\rm SL}_2(\C)_R\big)/{\pmb \mu}_{2,R}$ where ${\rm
SL}_2(\C)_R$ is the constant $R$-group scheme with
underlying (abstract) group ${\rm SL}_2(\C)$, and ${\pmb
\mu}_{2,R}$ embeds diagonally. We refrain, however, from
making such a conjecture.

\section*{Acknowledgments}

Z.C. gratefully acknowledges the support of University of
Alberta and of the Chinese Scholarship Council. He also
wishes to  thank H. Chen for helpful conversations. A.P.
wishes to thank the continued support of NSERC and CONICET.
Both authors wish to acknowledge the helpful comments by
the referees.

\noindent\textsc{Department of Mathematical and Statistical Sciences\\
University of Alberta\\
Edmonton, Alberta T6G 2G1\\
Canada} \\
\textit{E-mail address}: {\tt
zchang1@ualberta.ca}\\ \\
\textsc{Department of Mathematical and Statistical Sciences\\
University of Alberta\\
Edmonton, Alberta T6G 2G1\\
Canada\\
and\\
Centro de Altos Estudios en Ciencias Exactas\\
Avenida de Mayo 866, (1084)\\
Buenos Aires\\
Argentina}\\
\textit{E-mail address}: {\tt a.pianzola@ualberta.ca}\\
%\received{Received May 25, 2011}}

\begin{thebibliography}{99}

\bibitem{BN} G. Benkart and E. Neher, {\it The centroid of extended affine and root graded Lie algebras}, J.  Pure Appl. Algebra {\bf 205} (2006), 117--145.

\bibitem{CGP} V. Chernousov, P. Gille and A. Pianzola, {\it Torsors over the punctured affine line},
to appear in Am. J.  Math. (to appear)

\bibitem{D} M. Demazure,  {\it Sch\'emas en groupes r\'eductifs}, Bulletin de la Soci\'{e}t\'{e} Math. France {\bf  93} (1965), 369--413.

\bibitem{GP1} P. Gille and  A. Pianzola, {\it Isotriviality and \'{e}tale cohomology of Laurent polynomial rings}, J.  Pure Appl. Algebra {\bf 212}, (2008), 780--800.

\bibitem{GP2} P. Gille and A. Pianzola, {\it Isotriviality of
torsors over Laurent polynomials rings},  C. R. Acad. Sci. Paris,
Series I {\bf 340} (2005), 725--729.

\bibitem{GP3} P. Gille and A. Pianzola, {\it Galois cohomology and forms of algebras over Laurent polynomial rings},  Math. Ann.  {\bf 338} (2007), 497--543.

\bibitem{GP4} P. Gille and A. Pianzola, {\it Torsors, reductive group schemes and extended affine Lie algebras}, arXiv: 1109.3405v1, (2011).

\bibitem{KK} V. G. Kac, {\it Infinite dimensional Lie algebras}, 3rd ed., Cambridge University Press, 1990.

\bibitem{K} V. G. Kac, {\it Vertex algebras for beginners}, 2nd ed., American Mathematical Society, Providence, RI, 1998.

\bibitem{KLP} V. G. Kac, M. Lau and  A. Pianzola, {\it Differential conformal superalgebras and their forms}, Adv. Math. {\bf 222} (2009), 809--861.

\bibitem{P1} A. Pianzola, {\it Vanishing of $H^1$ for Dedekind rings and applications to loop algebras}, C. R. Acad. Sci. Paris, Series I {\bf 340} (2005), 633--638.

\bibitem{P2} A. Pianzola, {\it Derivations of certain infinite-dimensional algebras given by \'etale descent}, Math. Z. {\bf 264} (2010), 485--495.

\bibitem{SS} A. Schwimmer and  N. Seiberg, {\it Comments on the $N=2,3,4$ superconformal algebras in two dimensions}, Phys. Lett. {\bf B184}(2, 3) (1987), 191--196.

\bibitem{SGA1} A. Grothendieck, {\it Rev\^etements \'etales et groupe fondamental}, S\'eminaire de G\'eom\'etrie alg\'ebrique de l'I.H.E.S.,   Lecture Notes in Math. \textbf{224} Springer, 1971.

\bibitem{SGA3} M. Demazure and A. Grothendieck, {\it Sch\'emas en groupes}, S\'eminaire de G\'eom\'etrie alg\'ebrique de l'I.H.E.S., 1963--1964, Lecture Notes in Math. 151--153,
Springer 1970.
\end{thebibliography}
\end{document}